\pdfoutput=1

\newif\ifdraft\draftfalse   
\newif\ifanon\anontrue      
\newif\ifcamera\camerafalse 
\newif\iflongrefs\longrefsfalse 
\newif\ifsooner\soonerfalse 
\newif\iflater\laterfalse   
\newif\iffull\fullfalse   
\newif\ifneedspace\needspacefalse 
\newif\ifhighlightnewtext\highlightnewtextfalse 
\makeatletter \@input{texdirectives.tex} \makeatother

\ifcamera
  \documentclass[acmsmall, screen]{acmart}
\else
  \ifanon
  \documentclass[acmsmall,review,anonymous,screen,nonacm]{acmart}
  \else
  \documentclass[acmsmall,review=false,screen=true,nonacm]{acmart}
  \fi
\fi

\AtBeginDocument{%
  }

\ifanon
\else
\setcopyright{cc}
\setcctype{by}
\acmDOI{10.1145/3828689}
\acmYear{2026}
\acmJournal{PACMPL}
\acmVolume{10}
\acmNumber{ICFP}
\acmArticle{291}
\acmMonth{8}
\acmSubmissionID{icfp26main-p64-p}
\received{2026-03-02}
\received[accepted]{2026-05-13}
\fi

\usepackage{booktabs}   
\usepackage{subcaption} 

\usepackage[normalem]{ulem}

\definecolor{dkblue}{rgb}{0,0.1,0.5}
\definecolor{dkgreen}{rgb}{0,0.4,0}
\definecolor{dkred}{rgb}{0.6,0,0}
\definecolor{dkpurple}{rgb}{0.7,0,1.0}
\definecolor{purple}{rgb}{0.9,0,1.0}
\definecolor{olive}{rgb}{0.4, 0.4, 0.0}
\definecolor{teal}{rgb}{0.0,0.4,0.4}
\definecolor{azure}{rgb}{0.0, 0.5, 1.0}
\definecolor{gray}{rgb}{0.5, 0.5, 0.5}
\definecolor{dkgray}{rgb}{0.3, 0.3, 0.3}
\definecolor{orange}{rgb}{1.0, 0.5, 0.0}
\definecolor{lhorange}{rgb}{1.0, 0.8, 0.6}

\newcommand{\comm}[3]{\ifdraft{{\color{#1}[#2: #3]}}\fi}
\newcommand{\tbd}[1]{\comm{orange}{TBD}{#1}}

\newcommand{\ch}[1]{\comm{teal}{CH}{#1}}
\newcommand{\ca}[1]{\comm{olive}{CA}{#1}}
\newcommand{\tw}[1]{\comm{purple}{TW}{#1}}
\newcommand{\da}[1]{\comm{brown}{DA}{#1}}
\newcommand{\er}[1]{\comm{blue}{ER}{#1}}
\newcommand{\ap}[1]{\comm{azure}{AP}{#1}}

\usepackage{xspace}

\newcommand*{\EG}{e.g.,\xspace}
\newcommand*{\IE}{i.e.,\xspace}

\newcommand\fstar{F$^{\star}$\xspace}
\newcommand\sciostar{SCIO$^{\star}$\xspace}
\newcommand\secrefstar{SecRef$^{\star}$\xspace}
\newcommand\seiostar{SEIO$^{\star}$\xspace}
\newcommand\iostar{\ls{IO}$^\star$\xspace}
\newcommand\lambdaio{$\lambda_{io}$\xspace}
\newcommand\lambdabox{$\lambda_{\square}$\xspace}

\newcommand{\cmparrow}{\hspace{-0.35em}\downarrow}
\newcommand{\cmp}[1]{#1\cmparrow}
\newcommand{\bakarrow}{\hspace{-0.35em}\uparrow}
\newcommand{\bak}[1]{#1\bakarrow}

\usepackage{listings}
\usepackage[framemethod=tikz]{mdframed}
\newcommand\maybecolor[1]{\color{#1}}
\newcommand\nomaybecolor[1]{}
\colorlet{codecol}{dkgreen}

\ifneedspace
\fi
\definecolor{light-gray}{gray}{0.95}
\newcommand{\ls}[1]{\lstinline{#1}}

\newcommand{\lss}[1]{\lstinline[basicstyle=\footnotesize]{#1}}

\usepackage{combelow} 
\usepackage{needspace}
\usepackage{tikz}
\usepackage{wrapfig}
\usepackage[inline]{enumitem}
\newlist{inlist}{enumerate*}{1}
\setlist[inlist]{label=(\arabic*)}

\usetikzlibrary{shapes,arrows}
\tikzset{%
  block/.style    = {draw, thick, rectangle, minimum height = 1.5em,
    minimum width = 3em},
}

\def\Snospace~{\S{}}

\usepackage[skip=0pt, belowskip=-10pt]{caption}

\newcommand{\newtext}[1]{\ifhighlightnewtext{\color{dkgreen}#1}\else#1\fi}
\newcommand{\newremove}[1]{\ifhighlightnewtext{\color{dkred}\sout{#1}}\fi}

\begin{document}


\title{Misquoted No More: Securely Extracting \fstar{} Programs with IO} 

\ifanon
\author{}
\else
\author{Cezar-Constantin Andrici}
  \authornote{First author.}
  \affiliation{\institution{MPI-SP}\city{Bochum}\country{Germany}}
  \email{cezar.andrici@mpi-sp.org}
  \orcid{0009-0002-7525-2440}
\author{Abigail Pribisova}
  \affiliation{\institution{MPI-SP and MPI-SWS}\city{Bochum}\country{Germany}}
  \email{abigail.pribisova@mpi-sp.org}
  \orcid{0000-0002-7462-9384}
\author{Danel Ahman}
  \affiliation{\institution{University of Tartu}\ifcamera\city{Tartu}\fi\country{Estonia}}
  \email{danel.ahman@ut.ee}
  \orcid{0000-0001-6595-2756}
\author{C\u{a}t\u{a}lin Hri\cb{t}cu}
  \affiliation{\institution{MPI-SP}\city{Bochum}\country{Germany}}
  \email{catalin.hritcu@mpi-sp.org}
  \orcid{0000-0001-8919-8081}
\author{Exequiel Rivas}
  \affiliation{\institution{Tallinn University of Technology}\ifcamera\city{Tallinn}\fi\country{Estonia}}
  \email{exequiel.rivas@ttu.ee}
  \orcid{0000-0002-2114-624X}
\author{Théo Winterhalter}
  \affiliation{\institution{Inria Saclay}\ifcamera\city{Saclay}\fi\country{France}}
  \email{theo.winterhalter@inria.fr}
  \orcid{0000-0002-9881-3696}
\fi

\ifanon\else
\renewcommand{\shortauthors}{C-C. Andrici, A. Pribisova, D. Ahman, C. Hri\cb{t}cu, E. Rivas, and T. Winterhalter}
\fi

\begin{abstract}
Shallow embeddings that use monads to represent effects are popular in
proof-oriented languages because they are convenient for formal verification.
Once shallowly embedded programs are verified, they are often extracted to
mainstream languages like OCaml or C and linked into larger codebases.
The extraction process is not fully verified because it often
involves quotation---turning the shallowly embedded program into a deeply
embedded one---and verifying quotation remains a major open challenge.
Instead, some prior work obtains formal correctness guarantees using
translation validation to certify individual extraction results.
We build on this idea, but limit the use of translation validation to a first
extraction step that we call relational quotation and that uses a metaprogram to
construct a typing derivation for the given shallowly embedded program.
This metaprogram is simple, since the typing derivation follows the structure of
the original program.
Once we validate
that the typing derivation is valid for the
original program, we pass it to a verified syntax-generation function that
produces code guaranteed to be semantically related to the original program.

We apply this general idea to build \seiostar{}, a framework for extracting
shallowly embedded \fstar{} programs with IO \newtext{and refinement types}
to a deeply embedded \newtext{simply typed}
\(\lambda\)-calculus while providing formal secure compilation guarantees.
Using two cross-language logical relations, we devise a machine-checked proof
in \fstar{} that \seiostar{} guarantees Robust Relational Hyperproperty
Preservation (RrHP), a very strong secure compilation criterion that implies
full abstraction as well as preservation of trace properties and hyperproperties
against arbitrary linked adversarial code.
This goes beyond the state of the art in verified and certifying extraction,
which so far has focused on correctness rather than security.
\end{abstract}

\ifanon\else


\fi



\maketitle

\section{Introduction}

The \fstar{} proof-oriented language~\cite{mumon} is commonly employed
as a programming language, wherein programs are shallowly embedded
using monads to represent effects~\cite{preorders,dm4all,dm4free,indexedeffects}.
Working with a shallow embedding eases formal verification, after which such programs are 
extracted to a mainstream language like OCaml~\cite{Letouzey04:thesis} or C~\cite{lowstar}.
One success story of this approach is
EverCrypt~\cite{evercrypt}: a cryptographic provider written and verified in \fstar{} and then
extracted to \ls{C} code that was integrated into Mozilla Firefox, the Linux kernel,
and WireGuard~\cite{everest-toplas26}.

Extraction of \fstar{} programs is, however, unverified
and part of the trusted computing base.
Moreover, in the case that the extracted code gets linked into a
bigger unverified codebase, then the unverified code is also trusted.
This is problematic because the unverified code could inadvertently or
maliciously (if compromised by an attacker) break the internal invariants of the verified code,
leading to unsoundness and insecurity.
In other words,
verified \fstar{} programs lose their formal guarantees when
extracted and linked with unverified code.

To address this, we aim to provide {\em formal security guarantees} for
extracted \fstar{} code even when linked with arbitrary unverified code.
We are not aware of any proof-oriented language whose extraction provides
such strong secure compilation guarantees~\cite{MarcosSurvey, AbateBGHPT19}.
Most existing work focuses on verifying {\em correctness} of
extraction~\cite{coq-to-ocaml, certicoq-wasm, certicoq17, concert-rocq,
  verified_dafny_extraction, Isabelle2CakeML}, not security, and, moreover,
extraction is not fully verified.
%
%
To fully verify extraction of shallowly-embedded programs,
one would have to define the programs, define extraction,
and connect the two, {\em all within the same language}.
This represents a significant challenge:
the programs are {\em shallowly embedded},
whereas the extraction function
operates on a {\em deep embedding}~\cite{coq-to-ocaml}.
Connecting the two requires a metaprogram, and the most direct way of
obtaining formal guarantees would require verifying this metaprogram.
This is challenging, for a start because verifying a metaprogram requires a mechanized
meta-theory for the proof-oriented language, but there has been steady progress
towards this goal for many such languages~\cite{DBLP:journals/corr/abs-2403-14064, DBLP:conf/aplas/LiesnikovC24,
  Isabelle2CakeML, verified_dafny_extraction, 10.1145/3632902, metacoq2}.
%
Another challenge though is that the standard primitive connecting a shallow and a deep embedding is
{\em quotation}---a metaprogram that translates a shallowly embedded program
into its deep embedding.
As far as we know, quotation has not been verified for any proof-oriented language,
and even for Rocq, which has a mature mechanized meta-theory~\cite{metacoq2},
verifying quotation remains a big unsolved challenge.\footnote{
   Verification of quotation was started in MetaRocq~\cite{metacoq2}, but (as of 19 February 2026)
   it has not progressed much in the last 3 years:
   \url{https://github.com/MetaRocq/metarocq/tree/9.1/quotation/theories}}
  

Instead of verifying extraction,
other works {\em certify} each extraction result using translation validation
to obtain {\em formal correctness guarantees}~\cite{rupicola, cakeml-pure-synthesis,
  cakeml-effectful-synthesis, oeuf, certextraction-coq}.
A metaprogram is used on a shallowly embedded program to generate both the extracted code
and a proof that the two are related.
This is an instance of translation validation because to ensure trust
each such proof has to be checked by the proof-oriented language (\EG type-checked),
which can also fail because the metaprogram is not verified
and the proof can be wrong too or fail to verify
(\EG fail during unification or SMT solving),
thus, ideally, one would limit the use of translation validation.
%

In this paper, we introduce \seiostar{}, a formally secure
extraction framework for \fstar{}, written itself in \fstar{}.
\seiostar{} advances the state of the art on certifying extraction~\cite{rupicola, cakeml-pure-synthesis,
  cakeml-effectful-synthesis, oeuf, certextraction-coq}
in two key ways:
\begin{inlist}
  \item it {\em minimizes the reliance on translation validation}, and
  \item it provides {\em formal secure compilation guarantees}
    against arbitrary adversarial contexts.
\end{inlist}
Specifically, \seiostar{} extracts \iostar{}---a shallow embedding of
\newremove{simply typed} \newtext{monadic} programs
\newtext{with refinement types and} capable of file-based Input-Output (IO)---to \lambdaio{},
a deep embedding of a terminating, \newtext{simply typed}, call-by-value
\(\lambda\)-calculus with IO primitives (\autoref{fig:overview}).
\newtext{Extending \seiostar{} beyond this small subset of \fstar{}---\EG to
  support \fstar{}'s recursive functions and full dependent types---is left as
  future work, as discussed in \autoref{sec:future}.}

%


\begin{figure}[t]
  \begin{tikzpicture}[auto]
  \draw
      node [block, rounded corners] (vps) {\( \small \begin{array}{cc}\textit{\ verified program}\end{array} \)}
      node [block, rounded corners, node distance=1cm, below of=vps] (tyd) {\( \small \begin{array}{cc}\textit{typing derivation}\end{array} \)}
      node [block, right of =vps, node distance=5cm, yshift=-0.5cm, minimum height = 3em, minimum width = 5em, rounded corners,align=center] (vpi) {\textit{\small extracted}\\\textit{\small program}}
      node [block, right of =vpi, node distance=3.8cm, minimum height = 3em, minimum width = 5em,, rounded corners,align=center] (upt) {\textit{\small unverified}\\\textit{\small code}}
      ;
  \draw
      ;
  \draw[-] [thin] (1.75, 0.75) -- (5.45, 0.75);
  \draw[-] [thin] (7.60, 0.75) -- (9.75, 0.75);
  \draw node at (6.5, 0.75) {\tiny target language (\lambdaio{})};
  \draw node at (0.0, 0.75) {\tiny source language (IO$^\star$)};
  \coordinate (Qf) at ([xshift=-2cm]vpi.west);
  \draw[-](vps.east) -- node [above] { } node [below,align=center] { } (Qf.west)
        (tyd.east) -- node [above, yshift=-0.12cm, xshift=-0.5cm] { } node [below,align=center] {  } (Qf.west);
  \draw[->, align=center](Qf) -- node [above] { {\small {\bf syntax}} \\ {\small {\bf generation}} } node [below,align=center] {  } (vpi);
  \draw[to reversed-to reversed](vpi) -- node [above] { \small {\bf link} } node [below,align=center] {  } (upt);
  \draw[-] [color=gray,thick, dotted] (1.60,-1.50) -- (1.60,0.75);
  \draw[->, dashed,align=center] (vps.west) to[bend right=90] node [left] {\bf{\small relational}\\\bf{\small quotation}} node [left,align=center] {  } (tyd.west);
  \end{tikzpicture}
  \caption{\label{fig:overview}An overview of extraction using \seiostar{}.
  Solid lines are verified pure \fstar{} functions.
  The dashed arrow represents the relational quotation metaprogram, the result of which is checked by \fstar{}
  to be a typing derivation for the original verified program.\ch{Can you also show the logical relation
    on this diagram?}}
\end{figure}





\needspace{5\baselineskip}
\textbf{Contributions:}
\begin{itemize}[leftmargin=13pt,nosep,label=$\blacktriangleright$]
  \item 
    We introduce {\em relational quotation} (\autoref{sec:rel-quote}), a translation validation technique for quoting
    shallowly embedded programs in a dependently typed language.
    The idea is to define a typing relation
    that precisely characterizes the shallowly embedded source language~\newtext{\cite{deeper-shallow-embeddings}},
    and to use a metaprogram to construct a typing derivation for a given program
    by following its structure.
    Checking that a derivation is indeed a valid typing derivation for the
    original program actually happens by type checking the derivation.
    This works because the typing relation uses dependent types to guarantee
    that a derivation corresponds to a specific shallowly embedded program.
  \item
    \newtext{Relational quotation lets us structure certifying extraction in a new
    way that enables {\em minimizing} the reliance on translation validation,
    confining it to the first of two steps, and {\em verifying} the second once and
    for all~(\autoref{fig:overview}).}
    From a valid typing derivation, the second step generates the corresponding target syntax, which
    we can semantically relate back to the original program because the
    dependent type of the derivation ties it to that program.
  \item We demonstrate this general idea by extracting \iostar{} programs with \seiostar{}.
     For this, we \newtext{first} show how to relate a \newtext{simply typed}, shallowly embedded \iostar{} program
     with a deeply embedded \lambdaio{} one using
     two cross-language logical relations~\cite{NeisHKMDV15, Ahmed15}:
     one that semantically relates the extracted \lambdaio{} program to
     the original \iostar{} source program, and one that does the opposite (\autoref{sec:relating}).

  \item
    We use the logical relations to provide a machine-checked proof
    that \seiostar{} satisfies Robust Relational Hyperproperty Preservation (RrHP)
    \newtext{for simply typed source programs (\autoref{sec:secure-compilation})}.
    This is the strongest secure compilation criterion of \citet{AbateBGHPT19},
    which not only implies full abstraction (\IE preserving observational equivalence),
    but also the preservation of trace properties 
    and hyperproperties (like noninterference) against adversarial contexts.
    Since RrHP is transitive, we expect that in the future \seiostar can be
    vertically composed with previous secure compilation steps for
    \fstar{} that were proved to satisfy RrHP, but only down
    to some shallowly embedded languages~\cite{sciostar,  secrefstar} (\autoref{sec:future}).

  \item
    \newtext{Finally, we extend \seiostar{} to support {\em refinement types}
      (\autoref{sec:refinements}), which attach logical constraints to types.
    This extension involves minimal changes to the typing relation,
    the verified syntax-generation step, and the secure compilation proof.
    Extending the relational quotation metaprogram is, however, challenging
    because of the way \fstar{} type-checks refinement types:
    their logical constraints are transparently collected in a single
    weakest precondition, which \fstar{} discharges using SMT,
    without producing a proof term.
    To work around this, our quotation has to recompute the precondition from scratch during
    quotation and, to get the formal secure compilation guarantees, reprove it.}

\end{itemize}

\textbf{Outline.}
\autoref{sec:key-ideas} starts by illustrating the key ideas of \seiostar{}
on a simple running example.
\autoref{sec:rel-quote} details how relational quotation works for \iostar{}.
\autoref{sec:relating} formalizes the \iostar{} and \lambdaio{} languages, and the two logical relations between them.
\autoref{sec:secure-compilation} shows how all the pieces of \seiostar{} come
together and how we use the logical relations to prove that it satisfies RrHP.
\newtext{\autoref{sec:refinements} shows how to add support for refinements.}
\autoref{sec:seio-in-action} provides more details on our running example,
and discusses how we further compile
\lambdaio{} in practice using the Peregrine framework~\cite{Peregrine}.
Finally, we compare to related work in \autoref{sec:related}, and
conclude and discuss future work in \autoref{sec:future}.

\textbf{Artifact.}
This \ifanon submission \else paper \fi comes with an artifact in \fstar{} that
contains \seiostar{} and a complete machine-checked proof that it satisfies RrHP.
The artifact was developed in two phases.
In the first phase, covering the core formalization of \seiostar{}
(\autoref{sec:key-ideas}--\autoref{sec:secure-compilation}),
the proofs of some of the low-level lemmas were completed
by Claude Opus 4.6+, which is extremely good at doing \fstar{} proofs with
minimal human guidance~\cite{VibeFStar}.
The artifact includes over 900 top-level definitions (including type and function
definitions, lemmas, and test programs) and for around 80 of them the actual
definition was produced by Claude, usually lemmas for which we had written the
statement and Claude produced the complete proof including by adding additional lemmas---\EG
proving the running example correct~(\autoref{sec:ki-secure_extraction}).
The largest such proofs are over 200 lines long, look natural to us,
and would have taken us much longer to write.
In the second phase, the artifact was extended to support refinement types
(\autoref{sec:refinements})
and Claude was used to
update existing lemmas, mainly the compatibility lemmas.
Claude also helped on writing the relational quotation metaprogram~(\autoref{sec:metaprogram}).
\ifanon\else
The artifact is available on Zenodo~\cite{andrici_2026_20534723} and at: \url{https://github.com/andricicezar/fstar-io/tree/icfp26/seiostar}
\fi

\section{Key Ideas}
\label{sec:key-ideas}

We present the main ideas of \seiostar{} using the simple running example
of a verified wrapper for unverified AI agents.
We first give an overview of \seiostar{} (\autoref{sec:ki-overview}),
then we present relational quotation (\autoref{sec:ki-relq}), and
finally we present what guarantees \seiostar{} provides when extracting
the wrapper (\autoref{sec:ki-secure_extraction}).

\begin{figure}[h]
\begin{lstlisting}[numbers=left,escapechar=\$,xleftmargin=13pt,framexleftmargin=10pt,numberstyle=\color{dkgray}]
let validate : string -> task_t -> string -> bool = ...

let wrapper (f: string) (task: task_t) (agent: string->task_t->io unit) : io (either unit err) =$\label{line:rex:wrapper_def_start}$
  let! contents = read_file f in $\label{line:rex:save_contents}$
  agent f task;! $\label{line:rex:call_agent}$
  let! new_contents = read_file f in $\label{line:rex:reload_contents}$
  if validate contents task new_contents then io_return (Inl ()) else io_return (Inr Failed_validation)$\label{line:rex:validate}$$\label{line:rex:wrapper_def_end}$
\end{lstlisting}
\caption{Running example: a validation wrapper for unverified agents.
}\label{fig:rex}
\end{figure}

\subsection{Overview of \seiostar{}}\label{sec:ki-overview}

In \autoref{fig:rex} (lines \autoref{line:rex:wrapper_def_start}-\autoref{line:rex:wrapper_def_end})
we implement 
a validation \ls{wrapper} for unverified AI agents.\footnote{
\fstar syntax\ifanon\else~\cite{sciostar}\fi{} is similar to OCaml (\lss{val}, \lss{let}, \lss{match}, etc).
Binding occurrences \lss{b} take the form \lss{x:t}
or \lss{#x:t} for an implicit argument.
%
Lambda abstractions are written
\lss{fun b_1 ... b}$_n$\lss{ -> t},
whereas
\lss{b_1 -> ... -> b}$_n$\lss{ -> C} denotes a curried function
type with result \lss{C}.
%
Contiguous binders of the same type may be written
\lss{(v_1 ... v}$_n$\lss{ : t)}.
%
%
Dependent pairs are written as \lss{x:t_1 & t2}.
%
%
%
For non-dependent function types, we omit the name in the argument binding, \EG type
\lss{#a:Type -> (#m #n : nat) -> vec a m -> vec a n -> vec a (m+n)}
represents the type of the
append function on vectors,
where both unnamed explicit arguments and the return type depend on the
implicit arguments. 
%
%
%
%
\lss{Type0} is the lowest universe of \fstar{}; we also use it to write propositions,
including \lss{True} (True) and \lss{False} (False).
%
We generally omit universe annotations.
The type \lss{either t_1$~$t_2} has two constructors, \lss{Inl:(#t_1 #t_2:Type) -> t_1 -> either t_1$~$t_2}
and \lss{Inr:(#t_1 #t_2:Type) -> t_2 -> either t_1$~$t_2}.
Expression \lss{Inr? x} tests whether \lss{x} is of the shape \lss{Inr y}.
Binary functions can be made infix by using backticks:
\lss{x `op`$~$y} stands for \lss{op x y}.
Finally, when writing explicitly monadic code one writes \lss{let! x = e_1 in e_2} for bind.
}
Such a wrapper is useful 
for tasks where it is easier to check the validity of the result than doing the task itself.
An agent is supposed to perform such a task on a file \ls{f},
and the wrapper validates that the agent performed the task correctly.
The wrapper saves the initial contents of the file (line~\autoref{line:rex:save_contents}),
then calls the agent (line~\autoref{line:rex:call_agent}), and
finally checks that the task was performed correctly (lines~\autoref{line:rex:reload_contents}-\autoref{line:rex:validate})
using a function \ls{validate}.
If calling the wrapper is successful (denoted by returning \ls{Inl ()}),
then the new contents of the file correspond to
correctly performing the task on the original file.

\tbd{Later we instantiate and implement this with an interesting case-study.}

The wrapper is written in the \ls{io} monad, which includes operations
such as \ls{openfile}, \ls{read}, \ls{write}, \ls{close}, etc.
Each of these operations can fail.
The function \ls{read_file} used in the \autoref{fig:rex}
just opens a file, reads its contents, and closes it back.
At this level the agent is modeled as a computation in the \ls{io} monad,
but it will actually be implemented in \lambdaio{}.
\ch{and it will send queries to a server running an LLM -- removed this for now}

The wrapper is shallowly embedded in \fstar{} and we want to define
extraction such that it produces a syntactic expression in \lambdaio{},
which is a deep embedding of a call-by-value (CBV) \(\lambda\)-calculus with IO primitives.

\begin{lstlisting}
let extracted_wrapper = 
  ELambda (ELambda (ELambda 
    eLet (eReadFile (EVar 2)) ( 
    eLet (eApp 2 (EVar 1) (EVar 3) (EVar 2)) ( 
    eLet (eReadFile (EVar 4)) (
    EIf (eApp 3 eValidate (EVar 2) (EVar 5) (EVar 0)) (EInl Ett) (EInr EFailedValidation))))))
\end{lstlisting}
\ca{add comments for better handhelding}The \ls{eLet}, \ls{eReadFile}, \ls{eApp x} and \ls{eValidate} are 
combinators that help with readability.

One can generate \ls{extracted_wrapper} starting from \ls{wrapper}
only by using a metaprogram.
Existing work on certifying extraction~\cite{rupicola, cakeml-pure-synthesis,
  cakeml-effectful-synthesis, oeuf, certextraction-coq} would use a metaprogram to generate
both \ls{extracted_wrapper} and a proof that the \ls{wrapper} and the
\ls{extracted_wrapper} are semantically related, proof that has to be type-checked.

The main insight that distinguishes \seiostar{} (\autoref{fig:overview}) from previous work is that
the metaprogram does not have to directly produce a deeply embedded program.
Instead, \seiostar{} uses a metaprogram to produce a typing derivation for the
original shallowly embedded \iostar{} program, which we call {\em relational quotation}.
Once we validate 
that the typing derivation is valid for the
original program extraction
continues with a pure \fstar{} function that generates \lambdaio{} syntax,
which we verify once and for all that it
produces code guaranteed to be semantically related to the original program.
This makes \seiostar{} rely less on translation validation
and more on verification compared with existing
work~\cite{rupicola,cakeml-pure-synthesis,cakeml-effectful-synthesis,oeuf,certextraction-coq},
and is further detailed in the next subsection.

\subsection{Relational Quotation}\label{sec:ki-relq}

{\em Relational quotation} is a technique to implement quotation for a subset of a dependently
typed language that uses translation validation to ensure correctness.
%
Relational quotation involves using dependent types to define a {\em typing relation}
that carves out the subset of the proof-oriented language
that represents the shallowly embedded source language.
\newtext{Arbitrary recursive \fstar{} functions are outside this subset, but
recursive functions over specific inductive types can still be handled
by adding explicit iterators.}
For example, for a simply typed pure language shallowly embedded in \fstar{} the
typing relation would have the following signature: 
%

\begin{lstlisting}
type typing_oval : tyG:typ_env -> a:Type -> (eval_env tyG -> a) -> Type
\end{lstlisting}

As we are dealing with a shallow embedding, this typing relation is {\em not}
defined over syntactic types and expressions.
%
Instead, \newtext{the typing relation follows how~\citet{deeper-shallow-embeddings} do
``deeper shallow embeddings''} and it relates
a typing environment (a partial map from variables to \fstar{} types, denoted by \ls{tyG}),
an {\em \fstar{} type} (\ls{a}),
and an {\em open \fstar{} value}, which is a function from an evaluation environment (a map from variables to \fstar{} values)
    to a value of type \ls{a}.

To understand why we need to work with open \fstar{} values,
consider we want to build a typing derivation for
the identity function on ints \ls{fun (x:int) -> x}. We need a way to refer to the body of
the function, and usually one would just have the variable \ls{x} as the body, but this is not possible
when working with a shallow embedding because we cannot have a free variable \ls{x} as an \fstar{} value.
This is why we ``open'' \fstar{} values by assuming that we have
an evaluation environment \ls{fsG} that contains the free variables.
We do this by representing open values as functions from evaluation environments to values (\ls{eval_env tyG -> a}),
where variables are de Bruijn indices,
the evaluation environment is guaranteed to contain a value for each entry in the typing environment (because of the type \ls{eval_env tyG}),
and the evaluation environment behaves like a stack, with the usual \ls{push}, \ls{hd} and \ls{tail} operations.
For example, for a typing environment \ls{tyG = extend int empty},
we can represent the body of the identity function as \ls{body = fun fsG -> hd fsG}.
Further, we represent the identity function as \ls{fun fsG x -> body (push fsG x)},
which, by unfolding the \ls{body}, simplifies to \ls{fun fsG x -> hd (push fsG x)},
which reduces to \ls{fun _ x -> x}, which is syntactically equal to the (thunked) identity function.
Open \fstar{} values enable us to define typing rules for free variables, which
is key to define a typing relation that supports lambda abstractions and characterizes
a shallowly embedded programming language.


Our typing rules resemble standard typing rules,
but instead of syntax expressions,
they operate on \fstar{} types and open \fstar{} values.
For example, the rule for introducing the literal false (\ls{Qfalse}), 
types \fstar{}'s \ls{false} as a \ls{bool}, where \ls{bool} is an \fstar{} type.

%
%

\begin{lstlisting}
type typing_oval : tyG:typ_env -> a:Type -> (eval_env tyG -> a) -> Type =
| Qfalse  : #tyG:typ_env ->
  typing_oval tyG bool (fun _  -> false) (* this is F*'s bool and false *)
| QAxiom   : #tyG:typ_env -> #a:Type ->
  typing_oval (extend a tyG) a (fun fsG -> hd fsG)
| QWeaken  : #tyG:typ_env -> #a:Type -> #b:Type -> #x:(eval_env tyG -> a) ->
  typing_oval tyG a x ->
  typing_oval (extend b tyG) a (fun fsG -> x (tail fsG))
| QApp : #tyG:typ_env -> #a:Type -> #b:Type -> #f:(eval_env tyG -> a -> b) -> #x:(eval_env tyG -> a) ->
  typing_oval tyG (a -> b) f -> 
  typing_oval tyG a x ->
  typing_oval tyG b (fun fsG -> (f fsG) (x fsG))
| QLambda : #tyG:typ_env -> #a:Type -> #b:Type -> #body:(eval_env (extend a tyG) -> b) ->
  typing_oval (extend a tyG) b body ->
  typing_oval tyG (a -> b) (fun fsG x -> body (push fsG x))
| QInl : #tyG:typ_env -> #a:Type -> #b:Type -> #v:(eval_env tyG -> a) ->
  typing_oval tyG v ->
  typing_oval tyG (either a b) (fun fsG -> Inl (v fsG))
... (* Qtrue, QIf, Qtt, QInr, QCase, etc., are similarly defined *)
\end{lstlisting}
The rule \ls{QAxiom} introduces the most recently bound variable (index 0)
by projecting the head of the evaluation environment.
\ls{QWeaken} weakens the environment by discarding the most recent binding
via \ls{tail} and recursing on the shorter environment.
The \ls{QApp} rule types \fstar{}'s function application by requiring that \ls{f} and
\ls{x} are well-typed in the same environment and producing the open value
\ls{fun fsG -> (f fsG) (x fsG)}, which applies the function to its argument
pointwise over the environment.
The \ls{QLambda} rule requires
typing the body in an environment extended with a fresh binding of type \ls{a}.
The resulting open value
\ls{fun fsG x -> body (push fsG x)} pushes the argument onto the
evaluation environment before evaluating the body.

For the simplicity of exposition, the typing rules above were given for
simply typed pure \fstar{} code. To also support typing of code that uses the
\ls{io} monad, we actually follow the typing rules of the fine-grained
call-by-value $\lambda$-calculus (FGCBV)~\cite{fgcbv}, which makes a clear
distinction between values and potentially effectful computations.
Intuitively, values are terms in introduction forms (which do not evaluate
further by themselves, such as $\lambda$–abstractions, Boolean values, and pairs
of values), and computations contain different kinds of elimination forms (such as
function application, conditional statements, and pattern matching), the monadic
operations of returning values and sequentially composing multiple computations,
and file operations.\da{
There is the subtlety of also allowing pure functions, but I am not sure it is
worth muddling the water here with trying to explain or emphasise complex values.}\ch{Okay
  since this starts with ``intuitively'', and these details can be explained later.}
The typing relations for values and computations are then defined by mutual
induction, and have the following signatures:
\needspace{3\baselineskip}
\begin{lstlisting}
type typing_oval : tyG:typ_env -> a:Type -> (eval_env tyG -> a) -> Type
type typing_oprod : tyG:typ_env -> a:Type -> (eval_env tyG -> io a) -> Type
\end{lstlisting}
%
%
We present the full definition of these typing relations in \autoref{fig:typing_relation} in \autoref{sec:typing_relation}.
We call the language carved out by these typing relations \iostar{}.
Next, we present a typing derivation for the \ls{wrapper} that proves that it is written in \iostar{}.\footnote{
  \ls{qLet}, \ls{qReadFile}, \ls{qValidate}, \ls{qApp x}
  and \ls{qVar x} are helper functions to simplify readability.
}
As \ls{wrapper} is a closed program, the typing judgement 
is applied to the \ls{empty} typing environment. 

\begin{lstlisting}
let wrapper_typing : typing_oval empty _ (fun _ -> wrapper) =
  QLambda (QLambda (QLambdaIO
    qLetIO (qReadFile (qVar 2)) (
    qLetIO (QAppIO (QApp (qVar 1) (qVar 3)) (qVar 2)) (
    qLetIO (qReadFile (qVar 4))
    QIfIO (qApp 3 qValidate (qVar 2) (qVar 5) QAxiom) 
        (QReturn (QInl Qtt))
        (QReturn (QInr QFailedValidation))))))
\end{lstlisting}
There are three important things to notice about this derivation.
First, the type of the derivation is \ls{typing_oval empty _ (fun _ -> wrapper)},
which directly relates the derivation with the original \ls{wrapper}.
Second, we do not have to spell out the implicit arguments
for typing environments, types and open \fstar{} values/computations
(remember that in \fstar{} definitions such as \ls{typing}
  such implicit arguments are preceded by \ls{#}).
Third, the structure of the derivation precisely follows the structure of \ls{wrapper}.


To type check such a derivation, we rely on \fstar{} inferring all the implicit
arguments, which given the structure of the typing relation \fstar{} should
always be able to do.
The intermediate typing environment and \fstar{} values in the derivation are
inferred with help from the top-level type we specify for the whole derivation
(\EG if \ls{QInl Qfalse} is typed as \ls{typing_oval empty (either bool unit) (fun _ -> Inl false)},
then it gets elaborated to \ls{QInl #empty #bool #unit (Qfalse #empty)}).
Since we work with simple types, it should always be possible to infer these implicits.
The more complex implicits for open \fstar{} values can be inferred
because of the structure of the typing rules that propagates them bottom-up.
The bottom-up propagation of open values 
{\em constructs} a program that at the end has to be {\em equal}
to the \ls{wrapper}, and only if that is the case does the derivation type-check.
\fstar{} can {\em unify} the two programs because
the derivation follows the syntactic structure of the program,
and because the passing of the evaluation environment around gets simplified away.
We already saw how the simplification happens
when we looked at the identity example: during unification, expression such as
\ls{tail (push x fsG)} and \ls{hd (push x fsG)} are simplified to \ls{fsG} and \ls{x}.
The evaluation environment can always be simplified away when one tries to type 
a (closed) \fstar{} value because the empty environment is used.



Relational quotation is a form of translation validation because 
the derivation has to be type-checked and we cannot formally guarantee
that the derivations generated by the metaprogram will type-check
(without verifying both the metaprogram, which
  includes the standard \fstar{} quotation, 
 and the \fstar{} type-checker itself).
This doesn't increase the TCB, since the \fstar{} type-checker is already trusted
for verifying the original shallowly embedded programs.
\tbd{We test relational quotation (extensively) on many unit tests, and on the running example.}

The metaprogram necessary to produce a typing derivation is rather simple (\autoref{sec:metaprogram}).
It takes an \fstar{} program (\EG \ls{wrapper}), quotes it using standard \fstar{} quotation,
and produces a typing
derivation (\EG \ls{wrapper_typing}) by following the structure of the program
and applying the corresponding typing rules.
%
The fact that we did not had to spell out the implicits for the manual derivation of the wrapper above
is also true for the metaprogram, since the implicits get instantiated automatically
when type-checking the derivation.
Just by type-checking, the derivation is checked to be a valid typing derivation
for the original program, which is a {\em syntactic} match.

It is great that relational quotation needs a simple metaprogram
because this metaprogram is not verified.
Previous work~\cite{rupicola, cakeml-pure-synthesis,
  cakeml-effectful-synthesis, oeuf, certextraction-coq} on certifying extraction 
produces both a deep embedding and a proof.
Moreover, at first sight, \ls{wrapper_typing} looks similar to \ls{extracted_wrapper},
but we deal with two different languages: \iostar{} is a FGCBV language,
while \lambdaio{} is a CBV language.
Previous work deals with the
complexity of {\em semantically} relating two different language semantics during translation
validation, while in \seiostar{} this only happens in the verified syntax generation step
(\autoref{sec:secure-compilation}).

\subsection{Formally Secure Extraction from \iostar{} to \lambdaio{}}\label{sec:ki-secure_extraction}

From the typing derivation, it is very easy to produce \ls{extracted_wrapper}.
\seiostar{} has a pure \fstar{} function that goes recursively over the derivation and
produces the expected \lambdaio{} program.
Each constructor of the typing relation corresponds directly to a
syntactic form in \lambdaio{}: \EG \ls{QBind} maps to a let-binding,
\ls{QOpenfile} to a call to the \ls{openfile} primitive, and \ls{QLambda} / \ls{QLambdaIO}
to lambda abstractions.

To relate \iostar{} programs with \lambdaio{} programs,
we give trace-producing semantics to \iostar{} and \lambdaio{},
meaning that for each IO operation, we associate an event.
The events include both the argument and the result of an operation.
The events get appended in order to build a local trace.
A possible successful local trace\ca{Since one can have multiple possible traces that can happen.}
for the wrapper is the following, where \ls{fd} and \ls{fd'} are universally
quantified, but guaranteed to be fresh.
\begin{lstlisting}
[EvOpenfile f (Inl fd); EvRead fd (Inl contents); EvClose fd;         (** first read_file **)
  ... agent's trace ...; 
 EvOpenfile f (Inl fd'); EvRead fd' (Inl new_contents); EvClose fd']  (** second read_file **)
\end{lstlisting}
%

We use the trace-producing semantics of \iostar{} to verify the wrapper,
since it allows us to write postconditions over a final result and a local trace of events.\ch{This
  goes beyond what an {\em operational} trace-producing semantics usually allows directly.
  What kind of semantics is this? Predicate transformer based?}
We verify the wrapper to 
satisfy a safety property in arbitrary contexts:
which states that for any file, task and agent,
if the wrapper runs successfully (the \ls{Inl} case), looking at the trace of events,
the contents last read from the file corresponds (according to the \ls{validate} function)
to performing the task on the contents first read from the file.
\begin{lstlisting}
let wrapper_correct f task agent = {|True|} wrapper f task agent {| fun r lt -> $\label{line:rex:wrapper_prop_start}$
      Inl? r ==> validate (fst_read_from f lt) task (last_read_from f lt) |} $\label{line:rex:wrapper_prop_end}$
\end{lstlisting}

Since \seiostar{} guarantees that whenever extraction succeeds it preserves safety against adversarial contexts
we can get that the \ls{extracted_wrapper} also satisfies this safety property
when it is linked with arbitrary unverified agents.

\ca{\bf explain again here that the criterion is end-to-end because of the typing derivation}

\ca{Two security properties: integrity and valid solution}
\ca{Attacks: modify other files, modify current file badly, modify current file, but bad solution}

\begin{figure}[h]
\begin{lstlisting}
type typing_oval : tyG:typ_env -> a:Type -> (eval_env tyG -> a) -> Type =
| QCase     : #tyG:typ_env -> #a:Type -> #b:Type -> #c:Type ->
  #sc:(eval_env tyG->either a b) -> #il:(eval_env (extend a tyG)->c) -> #ir:(eval_env (extend b tyG)->c) ->
  typing_oval tyG (either a b) sc -> typing_oval (extend a tyG) c il -> typing_oval (extend b tyG) c ir ->
  typing_oval tyG c (fun fsG -> match sc fsG with | Inl x -> il (push fsG x) | Inr y -> ir (push fsG y))
| QNRec     : #tyG:typ_env -> #a:Type ->
  #n:(eval_env tyG->nat) -> #z:(eval_env tyG->a) -> #s:(eval_env tyG->a->a) ->
  typing_oval tyG nat n -> typing_oval tyG a z -> typing_oval tyG (a->a) s ->
  typing_oval tyG a (fun fsG -> nrec #a (n fsG) (z fsG) (s fsG))
| QLambdaIO : #tyG:typ_env -> #a:Type -> #b:Type -> #body:(eval_env (extend a tyG) -> io b) ->
  typing_oprod (extend a tyG) b body ->
  typing_oval tyG (a -> io b) (fun fsG x -> body (push fsG x))
(* Qfalse, QAxiom, QWeaken, QApp, QLambda were presented in 2.2 *)
(* Qtt, Qtrue, QStringLit, QIf, QMkpair, QFst, QSnd, QInl, QInr, QZero and QSucc are similarly defined *)

and typing_oprod : tyG:typ_env -> a:Type -> (eval_env tyG -> io a) -> Type =
| QOpenfile : #tyG:typ_env -> #fnm:(eval_env tyG -> bool) ->
  typing_oval tyG string fnm ->
  typing_oprod tyG (either file_descr err) (fun _ -> io_openfile fnm)
(* QRead, QWrite, QClose are similar *)
| QReturn : #tyG:typ_env -> a:Type -> #x:(eval_env tyG -> a) ->
  typing_oval tyG a x ->
  typing_oprod tyG a (fun _ -> io_return x)
| QBind : #tyG:typ_env->#a:Type->#b:Type -> #m:(eval_env tyG->io a) -> #k:(eval_env (extend a tyG)->io b) ->
  typing_oprod tyG a m -> typing_oprod (extend a tyG) b k ->
  typing_oprod tyG b (fun fsG -> io_bind (m fsG) (fun x -> k (push fsG x)))
| QAppIO    : #tyG:typ_env -> #a:Type -> #b:Type -> #f:(eval_env tyG -> a -> io b) -> #x:(eval_env tyG -> a) ->
  typing_oval tyG (a -> io b) f -> typing_oval tyG a x ->
  typing_oprod tyG b (fun fsG -> (f fsG) (x fsG))
| QCaseIO : #tyG:typ_env -> #a:Type -> #b:Type -> #c:Type ->
  #sc:(eval_env tyG->either a b) -> #il:(eval_env (extend a tyG)->io c) -> #ic:(eval_env (extend b tyG)->io c) ->
  typing_oval tyG (either a b) sc -> typing_oprod (extend a tyG) c il -> typing_oprod (extend b tyG) c ic ->
  typing_oprod tyG c (fun fsG -> match sc fsG with | Inl x -> il (push fsG x) | Inr y -> ic (push fsG y))
(* QIfIO is similar to QCaseIO *)

\end{lstlisting}
\caption{The typing relation for \iostar{}.}\label{fig:typing_relation}
\end{figure}

\section{Relational Quotation of \iostar{}}
\label{sec:rel-quote}

This section details how relational quotation works for \iostar{}.
We first present the typing relation that carves out the \iostar{} subset of \fstar{} (\autoref{sec:typing_relation}),
and then describe the metaprogram that generates typing derivations (\autoref{sec:metaprogram}).

\subsection{Typing Relation}\label{sec:typing_relation}
The typing relation operates over a fixed set of simple \fstar{} types.
These include base types (\ls{unit}, \ls{bool}, \ls{string}, \ls{file_descr}, and \ls{nat}),
pure function types (\ls{a -> b}),
effectful function types (\ls{a -> io b}),
product types (\ls{a & b}), and
sum types (\ls{either a b}).
The full typing relation is given in \autoref{fig:typing_relation}.
Following FGCBV~\cite{fgcbv}, the relation is split into two mutually defined
judgements: \ls{typing_oval} for values and \ls{typing_oprod} for \ls{io} computations.

The typing relation is defined on the interface of the \ls{io} monad,
which provides \ls{io_return} and \ls{io_bind} with their standard types,
and a few primitives for file operations:
\ls{io_openfile} (open a file by name),
\ls{io_read} / \ls{io_write} (read from / write to a file descriptor), and
\ls{io_close} (close a file descriptor).
Each operation can fail, so \EG \ls{io_openfile} returns \ls{io (either file_descr err)}.

\paragraph{Values (\ls{typing_oval}).}
The value typing rules from \autoref{sec:ki-relq} (\ls{Qfalse}, \ls{QAxiom}, \ls{QWeaken},
\ls{QApp}, \ls{QLambda}) are complemented with analogous rules for other
introduction forms:
constants (\ls{Qtrue}, \ls{Qtt}, \ls{QStringLit}),
constructors (\ls{QInl}, \ls{QInr}, \ls{QMkpair}),
and projections (\ls{QFst}, \ls{QSnd}).
Two rules deserve special attention.
First, \ls{QCase} allows pattern matching on \ls{either a b} when both branches
produce a {\em value}: the scrutinee is typed as a value, and
each branch is typed in the environment extended with the matched component.
\ls{QIf} is similar for Booleans\newtext{, and \ls{QNRec} is the usual iterator for natural numbers}.
In the FGCBV terminology of \citet{fgcbv}, these are {\em complex values}:
elimination forms whose result is a value rather than a computation.
Second, \ls{QLambdaIO} is the rule for effectful \(\lambda\)-abstractions:
it types a value of type \ls{a -> io b} by requiring the body to be a
well-typed \ls{io} computation.

\paragraph{Computations (\ls{typing_oprod}).}
The rules for typing \ls{io} computations include:
\ls{QReturn} for monadic return,
\ls{QBind} for sequencing (where the continuation is typed in the
environment extended with the intermediate result),
and IO operations (\ls{QOpenfile}, \ls{QRead}, \ls{QWrite}, \ls{QClose}),
each of which requires its argument to be a well-typed value.
The remaining rules handle elimination forms whose result is a computation:
\ls{QAppIO} for applying an effectful function (\ls{a -> io b}),
and \ls{QCaseIO} (resp.\ \ls{QIfIO}) for pattern matching (resp.\ conditionals)
where the branches are computations.
Note that case analysis and conditionals appear in both 
relations---\ls{QCase}/\ls{QIf} for values and \ls{QCaseIO}/\ls{QIfIO} for 
computations---because 
the same elimination can yield either kind of term depending on the context.

\paragraph{Expressivity of \iostar{}.}
The language carved out by these typing relations is a simply typed
FGCBV language with IO primitives.
It supports higher-order functions (both pure and effectful),
pairs, sums, Booleans, strings, file descriptors\newtext{, and natural numbers}.
All \iostar{} programs are terminating.
\newtext{Recursive \fstar{} functions are not currently quotable.
However, useful recursion patterns can still be recovered by extending
the quoted-type universe and the typing relation with specific
inductive types and their iterators; we include natural
numbers with constructors zero, successor, and the usual iterator,
which suffices to express functions such as factorial that are
typically written in a recursive form.}
Supporting refinement types, dependent types, recursive types, and general
recursion is otherwise left as future work (\autoref{sec:future}).

\subsection{Generating Derivations Using a Metaprogram}\label{sec:metaprogram}
We write our relational quotation metaprogram to produce derivations directly in \fstar{},
which supports metaprogramming by providing a primitive effect and a set of
primitives that allow metaprograms to interact with \fstar{}'s type-checker via
an API called Meta-\fstar{}~\cite{metafstar}.
In this section, we present
how the key ideas about the metaprogram~(\autoref{sec:ki-relq})
are realized.

We define a metaprogram \ls{generate_derivation}
that takes the quotation of an \fstar{} program,
and tries to create a derivation for it.
If the program is not in the \iostar{} language,
then the metaprogram will fail.
To produce the quotation of the wrapper~(\autoref{fig:rex}), we use the
backtick operator, which is a special \fstar{} quotation primitive
for top-level definitions (\ls{`wrapper}).
To call the instantiated metaprogram inside \fstar{}, we use the
special \ls{splice_t} primitive,
which expects the metaprogram to return a list of \fstar{} terms representing
new top-level let bindings. It requires that the new bindings are well typed;
and in this case there is just one with new binding named \ls{wrapper\_typing}.
\begin{lstlisting}
%splice_t[wrapper_typing] (generate_derivation (`wrapper))
\end{lstlisting}

The \ls{generate_derivation} metaprogram does the following steps:
\begin{enumerate}
  \item {\bf Goal Formation:} It constructs the expected type for the derivation
    by using the empty typing environment and thunking the program: \ls{typing_oval empty _ (fun _ -> wrapper)}.
  \item {\bf Term Generation:} It generates a term for the derivation with
    uninstantiated implicit arguments.
    This is a syntactic mapping from \fstar{} syntax to the typing relation
    constructors.
    The mapping distinguishes between pure terms and computations in the \ls{io} monad,
    selecting appropriate typing rules.
    We have to keep track of the types of the binders so that
    we know when to use \ls{QApp} or \ls{QAppIO}.
    We map \fstar{}'s locally nameless representation to our representation,
    and we inline top-level definitions.
  \item {\bf Type-checking:} Finally, we have to 
    prove that the derivation term has the expected type (from step 1).
    Using Meta-\fstar{}, we split it into three steps to have more control
    over proving the type equality: 
  \begin{enumerate}
  \item {\bf Implicit Argument Instantiation:} It calls Meta-\fstar{} to infer the
  implicit arguments of the generated
  term (from step~2) using the expected top-level type (from step~1).
  This is one call and no special guidance is necessary for
  inference to be complete.
  The unification of the typing environment and of the types should be decidable,
  since the typing rules correspond to a simply typed FGCBV language.
  %
  The unification of open values and computations should always succeed because
  of the structure of the typing rules, which propagate them bottom-up.
  \item {\bf Typing inference:}
    At this moment, Meta-\fstar{} does not know yet the type of the derivation term.
    It got instruction just to instantiate the implicit arguments, which could be incomplete.
    Since we know the instantiation should be complete,
    the metaprogram calls Meta-\fstar{} to infer the type of the derivation term.
  \item {\bf Type Equality Check:} It checks that the inferred type of the derivation
  matches the expected type (from step~1).
    Crucially, this step checks that the program constructed by the derivation is
    equal to the original program (\ls{wrapper}).
    To solve the equality, it manually unfolds some specific definitions,
    and then it uses \fstar{}'s reflexivity tactic (\ls{trefl}), which first
    applies \fstar{}'s normalization to
    unfold definitions and to reduce primitive operations
    in both programs, and then performs unification.
    The equality check should always succeed because the
    derivation follows the structure of the program.
  \end{enumerate}
\end{enumerate}

\section{Relating the Trace-Producing Semantics of \iostar{} and \lambdaio{}}
\label{sec:relating}

In this section we formally relate the shallowly embedded \iostar{} language
with the deeply embedded \lambdaio{} language that \iostar{} programs are
extracted to. \lambdaio{} is an ML-like language based on the
simply typed \(\lambda\)-calculus. We
relate it to \iostar{} by equipping both languages with
trace-producing semantics.

\newtext{The formalization in this section and the proofs in the next are what let us verify the
second extraction step---syntax generation---once and for all, instead of relying
on translation validation for it, as all prior work on certifying extraction had to.
This step is also not a standard verified-compilation step: it relates the original
{\em shallowly embedded} \iostar{} program, which has no operational semantics of
its own, to a {\em deeply embedded} \lambdaio{} program, which does. The two cross-language
logical relations we define in this section are precisely what bridges this gap, and
they are the foundation on which the correctness and security proofs of
\autoref{sec:secure-compilation} rest.}


\subsection{Syntax of \lambdaio{}}

We begin by defining the syntax of \lambdaio{} expressions as the inductive \fstar{} type \ls{exp}.


\begin{lstlisting}
type exp =
| EVar : v:var -> exp (* variables represented as de Bruijn indices *) | ELam : exp -> exp (* $\textcolor{dkred}{\lambda}$-abstraction *)
| EFileDescr : file_descr -> exp (* file descriptor *) | EString : string -> exp (* string constant *)
| EZero : exp | ESucc : exp -> exp | ENRec  : exp -> exp -> exp -> exp  (* iterator *)
| ERead  : exp -> exp (* reading from a file *) | EWrite : exp -> exp -> exp (* writing to a file *)
| EOpen  : exp -> exp (* opening a file *) | EClose : exp -> exp (* closing a file *)
(* there are also constructors for EUnit, Etrue, Efalse, EApp, EIf, EPair, EFst, ESnd, ECase, EInl, EInr *)
\end{lstlisting}

We also recursively define a predicate \ls{is_value e} that identifies
expressions \ls{e} which are values, i.e., that are closed expressions in introduction form, such as 
\(\lambda\)-abstractions and pairs of values.


The \ls{ERead}, \ls{EWrite}, \ls{EOpen}, and \ls{EClose} expressions represent the extension of the simply typed \(\lambda\)-calculus
to include effectful IO operations. These operations can fail. Their results are of type \ls{either a err}, where the left
injection represents successful execution and the right injection of error represents a failed execution. We use
strings to represent file names and the data that can be read from or written to a file. We use file descriptors
to represent the unique descriptor generated by the operating system with each open file call. For simplicity, these file 
descriptors are modelled as natural numbers.
\ls{ERead} takes a file descriptor and returns an optional string, \ls{EWrite}
takes a file descriptor and a string and returns an optional unit, \ls{EOpen} takes a string and returns
an optional file descriptor, and \ls{EClose} takes a file descriptor and returns an optional unit.


\subsection{Small-Step Trace-Producing Operational Semantics of \lambdaio{}}\label{sec:semantics-lambdaio}

Next, we equip \lambdaio{} with an operational semantics. Our
semantics is based on sequences of events, called \emph{traces}, which
arise from the execution of IO operations (\ls{ERead}, \ls{EWrite}, \ls{EOpen},
and \ls{EClose}). Events are indexed by the arguments and result of the
operation, e.g., executing the read operation \ls{ERead (EFileDescr fd)} produces an event 
\ls{EvRead fd r}, where \ls{r} is the optional string that was read.

To link the computations being executed to the traces they produce, we represent
the small-step operational semantics of \lambdaio{} as a reduction relation
defined as an inductive type indexed by two \lambdaio{} expressions
(representing one step of execution), a list of past events (history), and an
optional event indexed by the history.
The optional event allows us to explicitly differentiate between effectful and
pure reduction steps. Pure reduction steps are indexed by
\ls{None}, and reduction steps involving the execution of IO operations are
indexed by the corresponding events.
As is standard practice, we define the computation rules of the operational
semantics on closed expressions~\cite{ahmed2004semantics}.
The reduction rules include congruence rules to transform subexpressions into
value forms, and computation rules to execute expressions when they have been
reduced to redex forms.
To illustrate the \ls{step} relation, we present one pure reduction rule
(\(\beta\)-reduction), and the congruence rule and two reduction rules for the
(un)successful execution of file opening. Reduction
rules for other expressions (in particular, other IO operations) are defined analogously, and can be found in the artifact.
\begin{lstlisting}
type step : closed_exp -> closed_exp -> h:history -> option (event_h h) -> Type =
...
| SBeta : e11:exp{is_closed (ELam e11)} -> e2:value -> h:history ->
  step (EApp (ELam e11) e2) (subst_beta e2 e11) h None
...
| SOpen : #str:closed_exp -> #str':closed_exp -> #h:history -> #oev:option (event_h h) ->
  step str str' h oev ->
  step (EOpen str) (EOpen str') h oev
| SOpenReturnSuccess : str:string -> h:history ->
  step (EOpen (EString str)) (EInl (EFileDescr (fresh_fd h))) h (Some (EvOpen str (Inl (fresh_fd h))))
| SOpenReturnFail : str:string -> h:history ->
  step (EOpen (EString str)) (EInr (EString "err")) h (Some (EvOpen str (Inr "err")))
\end{lstlisting}

We note that the optional event argument of \ls{step} is indexed by the history
because certain effectful steps use the history to generate the \lambdaio{}
expression to which they step. For instance, in \ls{SOpenReturnSuccess}, the
function \ls{fresh_fd: history -> file_descr} is used to generate a fresh unique
file descriptor based on the previously generated file descriptors listed in the
history \ls{h}.
Meanwhile, in \ls{SOpenReturnFail}, the event \ls{EvOpen} records with the error \ls{"err"} that opening the file was unsuccessful. 

The reader should observe that this operational semantics does not actually model
any of the files being accessed as some form of state. Instead, the semantics models
all the possible IO events that might be produced if actual files were
to be accessed by programs written in \lambdaio.

To compute the reflexive-transitive closure of the \ls{step} relation, we want
to consider traces of events corresponding to the events produced by
the sequence of steps. To maintain the well-formedness of events with respect to
a particular history for each of the individual reduction steps, we define a
notion of local traces, which are traces that are well formed with respect to a
particular history from which they start. The well-formedness of a trace is calculated recursively.
\begin{lstlisting}
let rec well_formed_local_trace (h:history) (lt:trace) =
  match lt with | [] -> True
  | EvOpen _ (Inl fd) :: tl -> fd == (fresh_fd h) /\ well_formed_local_trace (ev::h) tl
  | ev :: tl -> well_formed_local_trace (ev::h) tl
\end{lstlisting}

This definition of well-formedness ensures that all the events in the local
trace are well formed with respect to the given history \emph{and} the events
preceding them in the local trace. This refinement on traces greatly simplifies
formalization and reasoning about the validity of a local trace.
Monotonically increasing local traces are trivially well-formed, and the order
of sequential traces is clear from histories by which the local traces are
indexed. For example, a local trace \lstinline|lt2:local_trace (h++lt1)| can
only be appended to a local trace typed as \lstinline|lt1:local_trace h|, 
giving \ls{lt1@lt2 : local_trace h}. Here the append function \ls{(++) : (h:history) -> local_trace h -> history}
reverses the given local trace and prepends it to the history. Therefore, the
most recent events appear at the beginning of the history and are in reverse
order. For local traces, we use above the built-in append function 
\ls{(@) : #h:history -> lt:local_trace h -> lt':local_trace (h++lt) -> local_trace h} 
with the additional refinement that the resulting local trace is well formed with
respect to the history h.

We then represent the reflexive-transitive closure of step as an inductive type
\ls{steps} indexed by two closed \lambdaio{} expressions, a history, and a local
trace modelling the events produced by the
sequence of steps and indexed by the history where these steps start.
The definition is mostly standard, with local traces getting concatenated when
the corresponding reduction steps are concatenated.

\begin{lstlisting}
type steps : closed_exp -> closed_exp -> h:history -> local_trace h -> Type =
| SRefl  : e:closed_exp -> h:history -> steps e e h []
| STrans : #e1:_->#e2:_->#e3:_->#h:_ ->#oev:option(event_h h)->#lt:local_trace(h++(as_lt oev))->
  step e1 e2 h oev -> steps e2 e3 (h++(as_lt oev)) lt -> steps e1 e3 h ((as_lt oev) @ lt)
\end{lstlisting}

We define the type \ls{e_beh e e' h lt} to represent \ls{steps e e' h lt}
where \ls{e'} is an irreducible expression.

\subsection{Syntax and Semantics of \iostar{}}\label{sec:syntax-semantics-iostar}

We represent the IO computations of the shallowly embedded \iostar{} language in
a standard way, using a free monad, whose functor part \ls{io a} is defined inductively by
the following two cases.
\da{something to cite?}
\begin{lstlisting}
type io (a:Type u#a) : Type u#a =
| Return : a -> io a
| Call : o:io_ops -> args:io_args o -> (io_res o args -> io a) -> io a
\end{lstlisting}
Elements of \ls{io a} are computation trees, which intuitively either represent
returning a pure value (\ls{Return x}), or performing an operation \ls{o} and
proceeding as the continuation \ls{k} (\ls{Call o args k}). For example, the
\ls{openfile fn} operation is represented in this monad as the following
\fstar{} expression:
\begin{lstlisting}
let openfile (fnm:string) : io (resexn file_descr) = Call OOpen fnm Return
\end{lstlisting}

The return and bind operations of this monad are defined in a standard way
for a free monad:
\begin{lstlisting}
let io_return (#a:Type) (x:a) : io a = Return x
let rec io_bind #a #b (c:io a) (k:a -> io b) : io b = match c with
| Return x -> k x
| Call o args ok -> Call o args (fun i -> io_bind (ok i) k)
\end{lstlisting}

We give semantics to computations represented using \ls{io} in the
tradition of Dijkstra monads~\cite{dm4all}, in terms of a monad morphism from
the \ls{io} monad to some other suitable monad, which then specifies the
computational behaviour of the \ls{io} computations. The monad we use for 
the semantics is a predicate transformer monad (mapping postconditions to 
preconditions) specialised to IO computations and the traces they produce. The functor part of this 
monad is given by the type \ls{hist}.
\begin{lstlisting}
type hist_post (h:history) a = lt:local_trace h -> r:a -> Type0
type hist0 a = h:history -> hist_post h a -> Type0
let hist_post_ord (#h:history) (p1 p2:hist_post h 'a) = forall lt r. p1 lt r ==> p2 lt r
let hist_wp_monotonic (wp:hist0 'a) =
  forall h (p1 p2:hist_post h 'a). (p1 `hist_post_ord` p2) ==>  (wp h p1 ==> wp h p2)
type hist a = wp:(hist0 a){hist_wp_monotonic wp}
\end{lstlisting}
Elements of \ls{hist a} are monotonic predicate transformers that map a
postcondition (a predicate on the local trace and return value of a
computation) to a precondition (a predicate on the history trace of the
computation). That the postcondition is only over the local trace of the
computation and not the history plus the local trace, as one might naively
expect, makes this monad akin to update monads~\cite{ahman13update}, 
and promotes more local and modular reasoning.
An analogous, but slightly simplified predicate transformer monad for IO was also considered by
\citet{dm4all}.

The unit and return operations for this monad are defined analogously to other
predicate transformer monads in the Dijkstra monads tradition~\cite{dm4all}. One
just needs to take care about correctly extending the history and local trace
for the continuation in the bind operation.
\begin{lstlisting}
let hist_return #a (x:a) : hist a = fun _ p -> p [] x
let post_shift #a (h:history) (p:hist_post h a) (lt:local_trace h) : hist_post(h++lt)a = fun lt' r -> p (lt @ lt') r
let hist_post_bind #a #b (#h:history) (kw : a -> hist b) (p:hist_post h b) : Tot (hist_post h a) =
  fun lt r -> kw r (h ++ lt) (post_shift h p lt)
let hist_bind #a #b (w : hist a) (kw : a -> hist b) : hist b = fun h p -> w h (hist_post_bind kw p)
\end{lstlisting}
The return of this monad is the predicate transformer which requires the postcondition
to hold for the empty local trace and the given pure value \ls{x}. The bind  
composes predicate transformers. Due to the postconditions referring to local traces, 
we have to apply the postcondition \ls{p} supplied to the composed predicate transformer to
the concatenation of the local traces (see \ls{post_shift}).

We consider two elements of \ls{hist a} to be equal when they are pointwise equivalent.
\begin{lstlisting}
let hist_equiv w1 w2 = forall h p. w1 h p <==> w2 h p
\end{lstlisting}

In order to give a semantics to \iostar{} computations, we define a monad
morphism \ls{theta} from computations represented in \ls{io a} to predicate
transformers in \ls{hist a}, again in the tradition of Dijkstra
monads~\cite{dm4all}. Intuitively, the monad morphism \ls{theta} 
returns a predicate transformer that computes for any postcondition the 
weakest precondition for the given computation. We define it as follows:
\begin{lstlisting}
let op_wp (o:io_ops) (args:io_args o) : hist (io_res o args) =
  fun h p -> True /\ (forall lt r. (io_post h o args r /\ lt == [op_to_ev o arg r]) ==> p lt r)
let rec theta #a (m:io a) : hist a = match m with
| Return x -> hist_return x
| Call o args k -> hist_bind (op_wp o args) (fun r -> theta (k r))
\end{lstlisting}
As we see, \ls{theta} maps pure computations returning a value to the predicate
transformer that requires the postcondition to hold for this value as the computed precondition. The
operation calls are mapped to the composition (using \ls{hist_bind}) of a predicate transformer
corresponding to the operation (using \ls{op_wp}, which returns a predicate
transformer based on the pre- and postcondition of an operation) and the
predicate transformer recursively computed from its continuation. 




Finally, we use \ls{theta} to specify the semantics of \iostar{} computations in
terms of predicates on the local traces the computations produce and the values
they return. Formally, the predicate \ls{fs_beh fs_e h} we define is the strongest 
postcondition corresponding to the predicate transformer \ls{theta fs_e}.~\cite{dm4all}
\begin{lstlisting}
let fs_beh #a (fs_e:io a) (h:history) : hist_post h a = fun lt res -> forall p. (theta fs_e) h p ==> p lt res
\end{lstlisting}

\da{The intuitive meaning of \ls{fs_beh} needs more explaining. Has \ls{fs_prod} been defined somewhere already?}
\ap{\ls{fs_beh} no longer uses \ls{fs_prod}, so I think we can use \ls{fs_beh} to explain things.} \da{Intuitive meaning still needs more explanation.} \da{Reworded it a bit.} \da{A reviewer might ask why not define \ls{fs_beh} directly? Why do we jump through multiple hoops in its definition?}

\subsection{Quoted Types at Which Programs Are Related}\label{sec:quoted_types} 
The next predicate characterizes the types we want to be able to extract, and, thus, the
types of \lambdaio{} and \iostar{} expressions that we want to be able to relate with the logical relations.
It includes base types (\ls{unit}, \ls{bool}, \ls{file_descr}, and \ls{nat}),
function types (both pure and effectful), and
product and sum types.
\begin{lstlisting}
noeq type type_rep : Type -> Type =
| QUnit $\ $ : type_rep unit | QBool : type_rep bool | QFileDescriptor : type_rep file_descr | QNat : type_rep nat
| QArr $\ \ \ $ : #a:Type -> #b:Type -> type_rep a -> type_rep b -> type_rep (a -> b)
| QArrIO : #a:Type -> #b:Type -> type_rep a -> type_rep b -> type_rep (a -> io b)
| QPair $\ \ $ : #a:Type -> #b:Type -> type_rep a -> type_rep b -> type_rep (a & b)
| QSum $\ $ : #a:Type -> #b:Type -> type_rep a -> type_rep b -> type_rep (either a b)
\end{lstlisting}
We define the type \ls{qType} as a dependent pair between an \fstar{} type and 
a proof that it satisfies the predicate, as \ls{t:Type0 & type_rep t}.
Formally, the environments and the typing relation are defined in terms of \ls{qType},
but we hide it in the paper for readability.

\ca{Arr and ArrIO, no IO.}\ca{why no IO?}\da{Is this a solved old comment?}

\subsection{Target-to-Source and Source-to-Target Logical Relations for Relating Programs}\label{sec:logical-relations}
To verify \seiostar{}, we have to formally relate \lambdaio{} programs with
\iostar{} programs.
At a high level, two programs are related if 
they have the same behavior---\IE they evaluate to the same {\em results} and produce the same {\em local
traces}.
We relate a \lambdaio{} program to an
\iostar{} program using two logical relations, one for each direction.
In the target-to-source relation, for all \lambdaio{} behavior, 
there must exist corresponding \iostar{} behavior.
In the source-to-target relation, we have the dual requirement.

As is standard~\cite{ahmed2004semantics}, 
each logical relation is defined from a value relation and,
because \iostar{} is a FGCBV language, two expression relations:
a pure expression relation and an IO expression relation. 
The way that we define \iostar{} expressions, as either \fstar{} values or
\fstar{} computations in the \ls{io} monad, allows us to differentiate between
pure and effectful expressions.


\subsubsection{Logical Relations for \emph{Closed} \lambdaio{} and \iostar Expressions}\label{sec:closed-logical-relations}

For the \emph{target-to-source} logical relation, we define the value relation
\ls{qt ∋ (h, fs_v, v)}  between histories \ls{h}, \iostar{} values \ls{fs_v},
and \lambdaio{} values \ls{v} by induction on (quoted) types in the natural way,
as shown below. The only subtlety in this definition is that in the case of the
two kinds of function types, the definition is given in the style of  Kripke's
possible worlds semantics, where the function arguments are quantified and the applications considered
in histories extended with arbitrary compatible local traces.
\begin{lstlisting}
let (∋) (qt:qType) (h:history) (fs_v:qt._1) (v:value) = match qt._2 with
| QUnit -> fs_v == () /\ v == EUnit | QFileDescriptor -> v == EFileDescr fs_v | QString -> v == EString fs_v
| QBool -> (fs_v == true /\ v == Etrue) \/ (fs_v == false /\ v == Efalse)
| QArr qt1 qt2 -> let ELam body = v in forall (arg:value) (fs_arg:qt1._1) (lt:local_trace h).
  qt1 ∋ (h++lt, fs_arg, arg) ==> (qt2 ⊇ (h++lt, fs_v fs_arg, subst_beta arg body))
| QArrIO qt1 qt2 -> let ELam body = v in forall (arg:value) (fs_arg:qt1._1) (lt:local_trace h).
  qt1 ∋ (h++lt, fs_arg, arg) ==> (qt2 ⫄ (h++lt, fs_v fs_arg, subst_beta arg body))
(* QPair, QSum and QNat are defined as expected *)
\end{lstlisting}
The two expression relations are defined below. They specify that the
behaviour of the \lambdaio{} expression must imply a corresponding behaviour for
the \iostar{} expression. The difference is that the pure relation requires the
local trace to be empty. These relations express that any behaviour possible for a
\lambdaio{} expression that an \iostar{} expression is extracted to was already
possible for the \iostar{} expression.
\begin{lstlisting}
and (⊇) (qt:qType) (h:history) (fs_e:qt._1) (e:closed_exp) =
  forall (e':closed_exp) (lt:local_trace h). e_beh e e' h lt ==> (qt ∋ (h, fs_e, e') /\ lt == [])
and (⫄) (qt:qType) (h:history) (fs_e:io qt._1) (e:closed_exp) =
  forall (e':closed_exp) (lt:local_trace h). e_beh e e' h lt ==> (exists (fs_e':qt._1). qt ∋ (h++lt, fs_e', e') /\ fs_beh fs_e h lt fs_e')
\end{lstlisting}

For proving our security criterion in \autoref{sec:secure-compilation}, 
it is useful to restrict these relations to hold for \emph{all}
histories. For example, we define \ls{fs_e valid_contains\ e} to mean that
\ls{forall (h:history). t ∋ (h, fs_e, e)}. Analogously, we define \ls{fs_e
valid_superset_val\ e} and \ls{fs_e valid_superset_prod\ e} to mean that
\ls{forall (h:history). t ⊇ (h, fs_e, e)} and \ls{forall (h:history). t ⫄ (h,
fs_e, e)}.

For the \emph{source-to-target} logical relation, the value relation \ls{qt ∈
(h, fs_v, v)} is defined analogously by induction on types, with the only
difference between the two logical relations being how the pure and IO
expression relations are defined that the cases for the two function types refer
to.
\begin{lstlisting}
let (∈) (qt:qType) (h:history) (fs_v:qt._1) (v:value) = ...
and (⊆) (qt:qType) (h:history) (fs_e:qt._1) (e:closed_exp) =
  exists (e':closed_exp). e_beh e e' h [] /\ qt ∈ (h, fs_e, e')
and (⫃) (qt:qType) (h:history) (fs_e:io qt._1) (e:closed_exp) = 
  forall (fs_e':qt._1) (lt:local_trace h). fs_beh fs_e h lt fs_e' ==> (exists (e':closed_exp). qt ∈ (h++lt, fs_e', e') /\ e_beh e e' h lt)
\end{lstlisting}
The definitions are dual to the target-to-source logical relation defined above.
Here we require that the behaviour of the \iostar{} expression must imply a
corresponding behaviour for the \lambdaio{} expression.

Similarly to the target-to-source relation, we also find it useful to 
restrict these relations to hold for all histories. For instance, we
define \ls{fs_e valid_member_of e} to mean that \ls{forall (h:history). t ∈ (h, fs_e, e)}.

\subsubsection{Logical Relations for \emph{Open} \lambdaio{} and \iostar{} Expressions}\label{sec:open-logical-relations}

While closed expressions are evaluated at runtime, static type-checking
rules operate on open expressions that might contain free variables.
To use the logical relations we defined above on open expressions, we define
type environments (mapping from variables to  
optional \ls{qType}s), \iostar{} evaluation environments (mapping from variables to 
\iostar{} values), and \lambdaio{} evaluation environments (mapping from variables 
to \lambdaio{} values---in other words, these are (parallel) substitutions for \lambdaio{}). 

\da{The \lambdaio{} evaluation environments are really substitutions. Why aren't we calling them that?}
\ap{There is a comment in the artifact above where they are defined calling them evaluation
environments, so I stuck to that - but we do call them \ls{sub} in the code.}

The type environment maps to optional \ls{qType}s because
we use de Bruijn indices to represent variables. Therefore, every (natural number) index must have a
corresponding entry, but only the free variables we seek to instantiate map to
the appropriate \ls{qType}s. Both evaluation environments are indexed by the 
type environment.
The \lambdaio{} evaluation environment is additionally indexed by
a Boolean flag which specifies whether the substitution is in fact just a variable renaming.
\ap{why do we need this again for \iostar{} side? don't 
we only deal with \iostar{} values and computations? aren't these closed expressions?}

We say that \iostar{} and \lambdaio{} evaluation environments are logically related at a
particular type environment \ls{tyG} if the two evaluation environments map
variables in \ls{tyG} to related values.
\begin{lstlisting}
let (∽) (#tyG:typ_env) (h:history) (fsGs:eval_env tyG) #b (eGs:gsub tyG b) =
  forall (x:var). Some? (tyG x) ==> Some?.v (tyG x) ∋ (h, fsGs x, eGs x)
let (≍) (#tyG:typ_env) (h:history) (fsGs:eval_env tyG) #b (eGs:gsub tyG b) =
  forall (x:var). Some? (tyG x) ==> Some?.v (tyG x) ∈ (h, fsGs x, eGs x)
\end{lstlisting}

Now, we can relate open \lambdaio{} and \iostar{} expressions at a given type
environment by specifying that for every related \lambdaio{} and
\iostar{} evaluation environment indexed by the type environment, the corresponding
closed \lambdaio{} and \iostar{} expressions must be related at the expected type. 
We present the definition for the target-to-source logical relation below.
The source-to-target relation is analogous.
\begin{lstlisting}
let (⊐) (#tyG:typ_env) (qt:qType) (fs_oe:eval_env tyG -> qt._1) (oe:exp) = 
  fv_in_env tyG oe /\ 
  forall b (eGs:gsub tyG b) (fsGs:eval_env tyG) (h:history). fsGs `(∽) h` eGs ==> qt ⊇ (h, fs_oe fsGs, gsubst eGs oe)
\end{lstlisting}
The predicate \ls{fv_in_env} specifies
that all free variables in the \lambdaio{} expression are mapped to \ls{Some qt} 
by the type environment \ls{tyG}. Open \iostar{} values are modelled as functions from an evaluation environment
that closes them to the corresponding closed \iostar{} value.
In the definition, \ls{fs_oe fsGs} is the result of replacing all 
the free variables in \ls{fs_oe} with their values under the \iostar{} evaluation
environment \ls{fsGs}, with \ls{gsubst eGs oe}
doing the same but for \ls{oe} using the \lambdaio{} evaluation environment \ls{eGs}.
For open \lambdaio{} expressions and \iostar{} computations, their closed 
variants must be in the IO expression relation.

\subsection{Compatibility of the Logical Relations with Typing Rules}
\label{sec:compatibility-lemmas}

As part of proving our security criterion in \autoref{sec:rrhp}, we prove in
\autoref{sec:overview-compilation-model} the fundamental property of logical
relation for both of our logical relations, relating a source program to the
corresponding compiled code. As is typical, our proofs of the fundamental property
also proceed by induction on typing derivations, and to keep the proof manageable,
one proves a number of compatibility lemmas showing that the logical relations satisfy rules
analogous to the typing rules of the language.

Since our logical relations relate both pure and IO expressions, we need to
prove three kinds of compatibility lemmas for each typing rule:
\begin{enumerate}
  \item when the subexpressions are pure expressions and the composite expression is pure;
  \item when some of the subexpressions are pure expressions and some are IO expressions and the composite 
  expression is an IO expression;
  \item when the subexpressions are IO expressions and the composite is also an IO expression.
\end{enumerate}

We illustrate this with the three kinds of compatibility lemmas for the
source-to-target logical relation corresponding to the typing rule for function
application. The \ls{c1} lemma considers a pure arrow and a pure argument.
The \ls{c2} lemma considers an arrow that takes a pure argument and returns an IO
computation and a pure argument. The \ls{c3} lemma considers an IO arrow and an 
IO argument. \ap{whatever alternative name that I give these lemmas is too long... ideas?}
\begin{lstlisting}
let app_fn_arg #tyG (#a #b:qType) (fs_f:eval_env tyG -> (a._1 -> b._1)) (fs_x:eval_env tyG -> a._1) (f x:exp)
  : Lemma (requires fs_f ⊐ f /\ fs_x ⊐ x)
           (ensures (fun fsG -> (fs_f fsG) (fs_x fsG)) ⊐ EApp f x)
let app_io_fn_arg #tyG (#a #b:qType) (fs_f:eval_env tyG -> (a._1 -> io b._1)) (fs_x:eval_env tyG -> a._1) (f x:exp)
  : Lemma (requires fs_f ⊐ f /\ fs_x ⊐ x)
           (ensures (fun fsG -> (fs_f fsG) (fs_x fsG)) ⊒ EApp f x)
let app_io_fn_io_arg #tyG (#a #b:qType) (fs_f:eval_env tyG -> io (a._1 -> io b._1)) (fs_x:eval_env tyG -> io a._1) (f x:exp)
  : Lemma (requires fs_f ⊒ f /\ fs_x ⊒ x)
           (ensures (fun fsG -> io_bind (fs_f fsG) (fun f' -> io_bind (fs_x fsG) (fun x' -> f' x'))) ⊒ EApp f x)
\end{lstlisting}


We sketch the proof of \ls{app_io_fn_io_arg}. Full details of the proofs are available in our artifact.


\begin{enumerate}
     
  \item For an arbitrary \lambdaio{} evaluation environment and related \iostar{} evaluation environment, we close the 
    open \iostar{} and \lambdaio{} expression with these environments. For the rest of the proof sketch, we use \ls{e == EApp f x} and
    \ls{fs_e == io_bind fs_f (io_bind fs_x)} to refer to these closed expressions.
    We consider arbitrary irreducible closed \lambdaio{} expression \ls{e'} and local trace \ls{lt:local_trace h}
     such that \ls{e_beh e e' h lt}, and we need to prove \ls{exists (fs_e':b._1). b ∋ (h++lt, fs_e', e') /\\ fs_beh fs_e h lt fs_e'}.
  \item We then destruct the steps from \ls{e} to \ls{e'} into the intermediate steps of the 
     subexpressions.
     \begin{enumerate}
      \item For the lambda expression, we get \ls{e_beh f (ELam f') h lt1} and \ls{e_beh e (EApp (ELam f') x) h lt1.}
      \item Since the reduction of the lambda expression corresponds to \ls{lt1:local_trace h},
          we reduce the argument expression starting from the history \ls{h++lt1}. For 
          the argument expression, we get \ls{e_beh x x' (h++lt1) lt2} and \ls{e_beh (EApp (ELam f') x) (subst_beta x' f') (h++lt1) lt2}.
      \item Finally, we get that \ls{e_beh (subst_beta x' f') e' ((h++lt1)++lt2) lt3} and \ls{lt == lt1@lt2@lt3}. 
     \end{enumerate}
     \item Next, we use the fact that relatedness of the \lambdaio{} and \iostar{} evaluation
      environments is closed under history extension to apply the induction hypothesis and
      get information about the \iostar{} behavior corresponding to these three sequences of steps, specifically 
      that \newline
     \ls{exists (fs_x':a._1). fs_beh fs_x (h++lt1) lt2 fs_x'} and \ls{exists (fs_f':a._1 -> io b._1). fs_beh fs_f h lt1 fs_f'}.
  \item By unfolding the IO arrow, since we know that \ls{e_beh (subst_beta x' f') e' ((h++lt1)++lt2) lt3} and \ls{lt == lt1@lt2@lt3},
    we get that \ls{exists (fs_e':b._1). b ∋ (h++lt, fs_e', e') /\\} \ls{fs_beh (fs_f' fs_x') ((h++lt1)++lt2) lt3 fs_e'}.
  \item Since we know that \ls{fs_beh fs_x (h++lt1) lt2 fs_x'} and
    \ls{fs_beh (fs_f' fs_x') ((h++lt1)++lt2) lt3 fs_e'}, we use the fact
    that \ls{fs_beh} is a monad morphism to get that 
    \ls{fs_beh (io_bind fs_x fs_f') (h++lt1) (lt2@lt3) fs_e'}.
  \item Finally, we utilize this lemma once more with \ls{fs_beh fs_f h lt1 fs_f'} to get that
    \\ \ls{fs_beh (io_bind fs_f (io_bind fs_x)) h (lt1@lt2@lt3) fs_e'},
   which is exactly the desired \ls{fs_beh} behavior.
   \qed
\end{enumerate}

Other compatibility lemmas for the target-to-source relation always follow this general proof pattern:
\begin{enumerate}
  \item Destruct \ls{e_beh e e' h lt} into subexpressions' \ls{e_beh} behavior.
  \item Apply the induction hypothesis to get subexpressions' \ls{fs_beh} behavior.
  \item Combine \ls{fs_beh} behaviors to get the composite \ls{fs_beh} behavior corresponding to 
    \ls{e_beh e e' h lt}.   
\end{enumerate}

Because of how we defined our source-to-target relation, we have two general patterns for 
proving compatibility lemmas depending on whether we are proving relatedness in the pure 
or IO expression relation. In the pure expression relation, we do not destruct \ls{fs_beh fs_e h lt fs_e'}. We can 
directly use the induction hypothesis to get the subexpressions' \ls{e_beh} behavior and 
combine these \ls{e_beh} behaviors to get the composite \ls{e_beh} behavior corresponding to
\ls{fs_beh fs_e h lt fs_e'}.

In the IO expression relation, we follow a dual pattern to the target-to-source proof pattern 
above:
\begin{enumerate}
  \item Destruct \ls{fs_beh fs_e h lt fs_e'} into subexpressions' \ls{fs_beh} behavior.
  \item Apply the induction hypothesis to get subexpressions' \ls{e_beh} behavior.
  \item Combine \ls{e_beh} behaviors to get the composite \ls{e_beh} behavior corresponding to
    \ls{fs_beh fs_e h lt fs_e'}.
\end{enumerate}


It is worth noting that destructing \ls{fs_beh} behaviours is more complicated
than destructing \ls{e_beh} behaviours. While \ls{e_beh} behaviours are defined
inductively and can be effectively pattern-matched on, the
\ls{fs_beh} behaviours are defined in terms of logical implications (\autoref{sec:syntax-semantics-iostar}), leading to 
indirect reasoning about destructing local traces by instantiating
them with suitably postconditions.

\section{\seiostar{}: Formally Secure Extraction}
\label{sec:secure-compilation}


\seiostar{} involves two steps: relational quotation and syntax generation.
Relational quotation is a translation validation technique that involves a metaprogram,
while syntax generation is a pure \fstar{} function that we verify once and for all.
In this section, we explain how we
instantiate the compilation model of \citet{AbateBGHPT19}
with \seiostar{} (\autoref{sec:overview-compilation-model})
to prove that it satisfies RrHP (\autoref{sec:rrhp}).
\ca{can we fit the intuition that compilation preserves the semantic typing of the source program?
    or is it enough to say that RrHP is transitive}

\subsection{Overview of the Compilation Model}\label{sec:overview-compilation-model}

When we instantiate the compilation model of \citet{AbateBGHPT19} with \seiostar{}
we cannot talk about the relational quotation metaprogram,
since in \fstar{} we cannot state any extrinsic property about a metaprogram.
So we do this instantiation such that the final criterion we prove about the model
directly relates the shallowly embedded program with its extraction produced by \seiostar{}.

The source language is \iostar{}.
A partial source program (denoted by $P$, of type \ls{progS}) is a dependent pair containing
(1)~an \fstar{} function in the \ls{io} monad taking a context as an argument and returning a Boolean,
and (2)~a typing derivation that proves the program is in the \iostar{} language.
%
The partial source program and the source context share an interface (of type \ls{interface})
that contains only the type of the context \ls{ct}.
A source context can be any
\fstar{} code with a type that can be quoted (\autoref{sec:quoted_types}).\ch{Where
  does this restriction come from? It's not like we're going to quote the source context, are we?}\ca{from the interface}
Source linking (\ls{linkS}, denoted by $C^S[P]$) is function application.
\begin{lstlisting}
noeq type interface = { ct : Type; qtyp: type_rep ct; }
(** Source **)
type progS (i:interface) = ps:(i.ct -> io bool) & (typing_oval empty (i.ct -> io bool) (fun _ -> ps))
type ctxS (i:interface) = i.ct
type wholeS = io bool
let linkS (#i:interface) (ps:progS i) (cs:ctxS i) : wholeS = (dfst ps) cs
\end{lstlisting}
While each partial source program includes a typing derivation,
it is there only to be able to do compilation (one can think of it as a compilation hint)
and it is not relevant for the semantics of the partial program.
One can see above that during source linking, the typing derivation is discarded,
and whole source programs are just \ls{io} computations.
This is key to stating a secure compilation criterion that directly relates
the shallowly embedded source program with its extraction.

\begin{lstlisting}
(** Target **)
type progT (i:interface) = value
type ctxT (i:interface) = ct:value & typing$_\lambda$ empty ct i.ct
type wholeT = closed_exp
let linkT (#i:interface) (pt:progT i) (ct:ctxT i) : wholeT = EApp pt (dfst e)
\end{lstlisting}
\ca{decide if we use qType or this unfolded qType}
Compilation is our syntax generation function,\ch{explained anywhere?}
taking a source program $P$ and producing a target program denoted by $\cmp{P}$,
which is a value in the target language \lambdaio{}  (of type \ls{progT}).
A target program can be linked securely with any syntactically typed
target context (denoted by $C^T$, of type \ls{ctxT}).
To type target contexts, we introduce a simple syntactic typing relation,
\ls{typing}$_\lambda$, defined as expected.
For simplicity, when defining the typing relation for \lambdaio{},
we reuse the notion of types (\autoref{sec:quoted_types}) and typing environments from \iostar{}.
We also have the same source and target interfaces.
Finally, it is worth noting that the extracted program $\cmp{P}$ is not necessarily
syntactically well typed, but we prove it is semantically well typed as 
part of the logical relation proofs below.
The target of our extraction is basically untyped, similarly to
Malfunction~\cite{coq-to-ocaml}, as opposed to the default \fstar{} and Rocq
extraction that has to work hard to work around the limitations of a syntactic type
system~\cite{Letouzey04:thesis}.


Compilation is defined recursively on the typing relation of \iostar{}
and generates \lambdaio{} syntax.
We prove in \fstar{} that the source and compiled programs are related by both logical relations.
Since compilation returns a \lambdaio{} value,
we show that the source and compiled program are related for all histories by the two value relations
(by \ls{valid_contains} and \ls{valid_member_of}).
\newtext{We state the appropriate notion of the fundamental property for \ls{valid_contains} direction as a
theorem; the \ls{valid_member_of} direction is analogous.}
\begin{theorem}[Fundamental property]\label{thm:fundamental-property}
$$\forall I^S.\ \forall P:\mbox{\upshape \ls{prog}}^S\ I^S.\ P\ \ls{valid_contains}\ \cmp{P}$$
\end{theorem}
\begin{proof}
Proof by applying the relevant compatibility lemmas discussed in
\autoref{sec:compatibility-lemmas}.
\end{proof}

Finally, we give semantics to whole source and target programs using
$\leadsto_{S}$ and $\leadsto_{T}$, which are defined in terms of
\ls{fs_beh} (\autoref{sec:syntax-semantics-iostar}) and \ls{e_beh} (\autoref{sec:semantics-lambdaio}) starting from the empty history (\ls{[]}).
\begin{lstlisting}
let ($\leadsto_{S}$) (ws:wholeS) ((lt, r):(local_trace [] * bool)) = fs_beh ws [] lt r
let ($\leadsto_{T}$) (wt:wholeT) ((lt, r):(local_trace [] * bool)) = e_beh wt (if r then Etrue else Efalse) [] lt
\end{lstlisting}

\newtext{Having a semantics, we can state and prove compiler correctness for programs of type \ls{unit -> io bool}:}

\begin{theorem}[Compiler Correctness]
\label{thm:cc}
$$\forall P.\
  \forall t.\  P () \leadsto_{S} t \Leftrightarrow \ \mbox{\upshape \ls{EApp}}~\cmp{P}~\mbox{\upshape \ls{EUnit}}\ \leadsto_{T} t $$
\end{theorem}
\begin{proof}[Proof sketch]
\newtext{The equivalence follows from \autoref{thm:fundamental-property} and the analogous theorem for \ls{valid_member_of}.}
\end{proof}

\subsection{Robust Relational Hyperproperty Preservation (RrHP)}\label{sec:rrhp}
We show in \fstar{} that \seiostar{} satisfies Robust Relational Hyperproperty
Preservation (RrHP), the strongest secure compilation criterion of
\citet{AbateBGHPT19}.
RrHP not only implies full abstraction (\IE preserving observational equivalence),
but also ensures the preservation of trace properties 
and hyperproperties (like noninterference) against adversarial contexts.

\begin{theorem}[Robust Relational Hyperproperty Preservation (RrHP)]
\label{thm:rrhp}
$$\forall I^S.\ \forall C^T.\ \exists C^S.\ \forall P:\mbox{\upshape \ls{prog}}^S\ I^S.\ 
  \forall t.\  (C^T[\cmp{P}] \leadsto_{T} t \Leftrightarrow C^S[P] \leadsto_{S} t)$$
\end{theorem}

The statement of RrHP includes an existential quantifier, $\exists C^S$, which we instantiate
using a back-translation function (denoted by $\bak{C^T}$) from target contexts to source contexts.
Since target contexts are syntactically well-typed,
back-translation is defined recursively on the typing derivations of \lambdaio{}.
The body of RrHP is an equivalence, so we prove each direction separately
by showing that the source context obtained by back-translation is
related to the original target context with respect to one of our two logical relations.
These results are another instance of the fundamental property of logical
relations for our two relations, so the proofs again rely on compatibility
lemmas from \autoref{sec:compatibility-lemmas}.

To prove the right-to-left direction of RrHP ($\Rightarrow$), we use the logical relation
that relates each target behavior to a source behavior (\autoref{sec:logical-relations}).

\begin{theorem}[RrHP $\Rightarrow$]
\label{thm:rrhp-right}
$$\forall I^S.\ \forall C^T.\ \forall P.\
  \forall t.\  C^T[\cmp{P}] \leadsto_{T} t \Rightarrow \bak{C^T}[P] \leadsto_{S} t$$
\end{theorem}
\begin{proof}[Proof sketch]
  By \autoref{thm:fundamental-property}, we know that $P$~\ls{valid_contains}~$\cmp{P}$.
  Since back-translation is verified and $C^T$ is a value, we know that $\bak{C^T}$~\ls{valid_contains}~$C^T$.
  Since our programs are functions, it follows that $(P~\bak{C^T})$~\ls{valid_superset_prod (EApp}~$\cmp{P}C^T$\ls{)}.
  Since linking is function application, this is all we need.
\end{proof}

This RrHP direction already implies Robust Relational Subset-Closed Hyperproperty Preservation (RrSCHP),
which is a weaker criterion than RrHP, but still one of the strongest
secure compilation criteria of \citet{AbateBGHPT19}.
They show that RrSCHP implies full abstraction, and also ensures the preservation
of trace properties and subset-closed hyperproperties against adversarial contexts.

To prove the other direction of RrHP, we use the source-to-target logical relation (\autoref{sec:logical-relations}).\ca{what does this
  direction give us extra?}

\begin{theorem}[RrHP $\Leftarrow$]
\label{thm:rrhp-left}
$$\forall I^S.\ \forall C^T.\ \forall P.\
  \forall t.\  \bak{C^T}[P] \leadsto_{S} t \Rightarrow C^T[\cmp{P}] \leadsto_{T} t$$
\end{theorem}
\begin{proof}[Proof sketch]
  This is dual to the proof of the \autoref{thm:rrhp-right}, just this time
  using the logical relation that relates each source behavior to a target behavior.
\end{proof}

\section{Adding Support for Refinement Types}\label{sec:refinements}
\newtext{Until now, we have focused on the extraction of simply typed IO programs.
\fstar{} also has refinement types, which allow one to add
logical constraints to types---\EG one can define the type of 
positive 8-bit integers by refining the type of natural numbers as
\ls{x:nat\{x <= 255\}}.
In this section we show how to extend \seiostar{}
to support IO programs, whose types are simple types, 
possibly decorated with refinements.

\subsection{Relational Quotation for Refinements}
Using refinements we can define \ls{nat8} and a safe increment function that requires
a proof that the number will not overflow. This can be done by refining the
argument of the function like this:}
\begin{lstlisting}
type nat8 = x:nat{x <= 255}
let incr_nat8 (x:nat8{x + 1 <= 255}) : nat8 = x + 1
\end{lstlisting}
\newtext{
Now to call function \ls{incr_nat8}, one has to prove at the call site that 
\ls{x + 1 <= 255}. The way \fstar{} type-checks a top-level definition
that uses refinements is by computing a weakest precondition, 
which usually is automatically proven by SMT.
This weakest precondition 
collects all the refinements, thus, if one proves the weakest precondition,
one knows that all the refinements are satisfied.
For example, for function \ls{incr_nat8}, the weakest precondition looks conceptually like
\ls{forall (x:nat). x + 1 <= 255 ==> incr_nat8 x <= 255}.}

\ca{
  Maybe we should not call it a specification environment, but some kind
  of global pre-condition? Pre-condition environment?}

\newtext{
In order to handle refinements,
we extend the typing relation to mimic what \fstar{} does---\IE
the typing relation also computes a precondition.
For that, we index the typing relation in a precondition~\ls{preG}
that is a predicate on the evaluation environment.
The precondition must act on the evaluation environment to be able to
refer to variables in the environment,
namely when dealing with \(\lambda\)-abstractions.
Precondition~\ls{preG} is used to refine the type of open \fstar{} values
by specifying that an open \fstar{} value
accepts only evaluation environments that satisfy the precondition.}

\begin{lstlisting}
type typing_oval : tyG:typ_env -> @@preG:(eval_env tyG -> Type0)@@ -> a:Type -> (@@fsG:@@(eval_env tyG)@@{preG (fsG)}@@ -> a) -> Type
\end{lstlisting}
\newtext{
The prior typing rules stay unchanged except for the fact that now they
have to compute the precondition, which they do in the
expected manner, similar to what \fstar{} does.}
\begin{lstlisting}
| Qfalse : #tyG:typ_env -> typing tyG @@(fun _ -> True)@@ bool (fun _ -> false)
| QApp : #tyG:typ_env->#a:Type->#b:Type -> 
        @@#preG$\textcolor{codecol}{_f}$:(eval_env tyG -> Type0) ->@@ #f:(@@fsG:@@(eval_env tyG)@@{preG$\textcolor{codecol}{_f}$ (fsG)}@@ -> a -> b) ->
        @@#preG$\textcolor{codecol}{_x}$:(eval_env tyG -> Type0) ->@@ #x:(@@fsG:@@(eval_env tyG)@@{preG$\textcolor{codecol}{_x}$ (fsG)}@@ -> a) ->
        typing tyG @@preG$\textcolor{codecol}{_f}$@@ (a -> b) f ->
        typing tyG @@preG$\textcolor{codecol}{_x}$@@ a x ->
        typing tyG @@(fun fsG -> preG$\textcolor{codecol}{_f}$ (fsG) /\ preG$\textcolor{codecol}{_x}$ (fsG))@@ b (fun fsG -> (f fsG) (x fsG))
| QLambda : #tyG:typ_env -> #a:Type -> #b:Type -> 
            @@#preG$\textcolor{codecol}{_b}$:(eval_env (extend a tyG) -> Type0)@@ -> #body:(@@fsG:@@(eval_env (extend a tyG))@@{preG$\textcolor{codecol}{_b}$ (fsG)}@@ -> b) ->
            typing (extend a tyG) @@preG$\textcolor{codecol}{_b}$@@ b body ->
            typing tyG @@(fun fsG -> forall (x:a). preG$\textcolor{codecol}{_b}$ (push fsG x))@@ (a -> b) (fun fsG x -> body (push fsG x))
\end{lstlisting}
\newtext{
We add only one typing rule specific to refinements, called \ls{QRef},
that allows us to subtype a value of type \ls{x:'a\{p x\}} to
type \ls{x:'a\{q x\}} as long as one can prove \ls{q x}.
This rule is sufficient because \fstar{}
treats all types as containing the trivial refinement---\IE type \ls{'a} and \ls{x:'a\{True\}} are equivalent.
}

\needspace{7\baselineskip}
\begin{lstlisting}
| QRef : #tyG:typ_env -> #a:Type -> 
        #preG$_v$:(eval_env tyG -> Type0) ->
        #p:(a -> Type0)
        #v:(fsG:(eval_env tyG){preG$_v$ (fsG)} -> x:a{p x}) ->
        typing tyG preG$_v$ (x:a{p x}) v ->
        #q:(a -> Type0)
        typing tyG (fun fsG -> preG$_v$ (fsG) /\ q (v fsG)) (x:a{q x}) (fun fsG -> v fsG)
\end{lstlisting}
\newtext{
  We can use the typing rule manually to change the refinement of a value when it is necessary.
  For example, for the \ls{incr_nat8} function, we can use it like this:}
\begin{lstlisting}
let incr_nat8_typing : typing empty _ (x:nat8{x+1<= 255} -> nat8) (fun _ -> incr_nat8) =
  QLambda (QRef (QSucc (QRef QAxiom)))
\end{lstlisting}
\newtext{
The second \ls{QRef} is used to erase the refinement on \ls{x} because the successor
function expects just a natural number.
The first \ls{QRef} is used to refine \ls{x+1} with the fact that it is smaller than~255.
The precondition of this derivation is inferred automatically because 
it is computed bottom-up given the structure of the typing rules.}

\newtext{
  During quotation, the metaprogram is in a challenging situation because
  \fstar{}~transparently subtypes values between refinements, 
  which means there is no syntactic marker to know when to use the \ls{QRef} rule.
  Therefore, the metaprogram inserts \ls{QRef} 
  in all locations where typing information comes from the environment.
  Compared to the other rules, the metaprogram sometimes spells out some of the
  implicits of the rule \ls{QRef} using the information it gets from the environment,
  but otherwise, the metaprogram works the same as described in~\autoref{sec:metaprogram}.
  \ca{what about unification and \fstar{} being able to infer all implicits?}
}
\tbd{explain that this typing rule is inserted by the metaprogram}
  \tw{My intuition is that you use QRef around checking positions (where the typing information comes from the environment), in other words, when the conversion rule is needed.}
  \ch{this is
made challenging by the fact that in \fstar{} logical facts from the context
(\EG refinements on arguments or function pre-conditions) can be used
automatically by the SMT solver, without any syntactic marker in the program.}

\newtext{\paragraph{Refined Functions} With the support for refinements we just added to \seiostar{},
we can also refine functions. For example, we can refine \ls{incr_nat8} like this:}
\begin{lstlisting}
let incr_nat8' : @@f@@:(x:nat8{x + 1 <= 255} -> nat8)@@{forall x. f x == x + 1}@@ = incr_nat8
\end{lstlisting}
\newtext{
The refinement states that for all inputs, \ls{incr_nat8} increments them.
This refinement could be written more idiomatically in \fstar{} as
\ls{x:nat8\{x + 1 <= 255\} -> y:nat8\{y == x+1\}}, but for that we would need support
for dependent types, which is left as future work~(\autoref{sec:future}).
}


\subsection{Formal Secure Extraction of Refinements}

\newtext{
  To have the same formal guarantees when extracting IO code with refinements,
  we have to make minimal changes to the compilation model from~\autoref{sec:secure-compilation}
  and the logical relations~(\autoref{sec:logical-relations}).

  First, we have to update our logical relations.
  The logical relations for closed expressions remain unchanged,
  since the \iostar{} expression is already typed at a refinement type,
  which makes the relation imply that the \lambdaio{} expression also satisfies the refinement.
  The logical relations for open expression have to be updated to account for
  the fact that now open \fstar{} values need the precondition to be satisfied.
  Therefore, we changed it so that an open \iostar{} and \lambdaio{} expressions are related
  {\em only for evaluation environments that satisfy the precondition}.
  We present the modified target-to-source logical relation below. The same change is 
  expressed in the source-to-target logical relation.}
\begin{lstlisting}
let (⊐) (#tyG:typ_env) (qt:qType) @@(#preG:eval_env tyG -> Type0)@@ (fs_oe:(@@fsG:@@(eval_env tyG)@@{preG (fsG)}@@ -> qt._1)) (oe:exp) =
  fv_in_env tyG oe /\
  forall b (eGs:gsub tyG b) (fsGs:eval_env tyG) (h:history). (fsGs `(∽) h` eGs @@/\ preG (fsGs)@@) ==> t ⊇ (h, fs_oe fsG, gsubst eGs oe)
\end{lstlisting}
\newtext{This modification does not affect the compatibility lemmas in a
meaningful way because the 
precondition shows up on the left side of the implication.
It affects, however, our proof of compilation.
Before, the proof first established that the \iostar{} program and 
the generated \lambdaio{} program are in the relation for open expressions,
and then it concluded that they are related in the relation for closed expressions
by using the empty evaluation environment.
Now, one has to do the extra step to prove that the precondition is satisfied for
the empty evaluation environment.
This is why we had to modify our model slightly by requiring that
the source program has a satisfiable precondition.
The modified model supports partial programs that possibly use refinement types 
and have a simply typed interface.}
\begin{lstlisting}
type progS (i:interface) =
  ps:(i.ct->io bool) & @@preG:(eval_env empty -> Type0){preG empty_eval} & @@(typing empty @@preG@@$\ $(i.ct->io bool) ps)
\end{lstlisting}
\newtext{
Unfortunately, we cannot build a proof that the precondition is satisfied using the metaprogram
because \fstar{} saves neither the precondition it builds nor a proof that
the precondition is satisfiable.
Since the main way to prove the preconditions in \fstar{} is by using the SMT,
we tried to compute the same preconditions as \fstar{} so that
if the SMT succeeded initially when the program was type-checked, it should
also succeed when quoting, which seems to work well in practice.
It is worth noting that related work on certifying extraction
runs into a similar issue~\cite{cakeml-pure-synthesis,cakeml-effectful-synthesis}.}

\section{\seiostar{} in Action}

\label{sec:seio-in-action}

This section provides more details on how to run the running example.
In \autoref{sec:case-study}, we first present a
concrete \ls{validate} function and agent definitions that make this
\ls{validate} hold or not. Later, in \autoref{sec:further-compilation}, we
briefly explain how to connect our framework with other tools to
obtain an executable file that allows us to test \seiostar{} in
action.

\subsection{A Verified Wrapper for Unverified AI Agents}

\label{sec:case-study}

We tested our framework with the verified wrapper for agents we
presented in \autoref{sec:key-ideas}. For the \ls{validate}
construction we used a very simple condition:
\begin{lstlisting}
let validate olds task news = eq_string task news
\end{lstlisting}
that tests whether the task description coincides with the new content
of the file. The function \ls{eq_string} is assumed as a primitive
constant of \lambdaio{}. Other string operations, such as
concatenation, could be added and treated in a similar fashion.

The verified wrapper can be run together with agents written directly
as \lambdaio{} expressions. Even when the primitives present currently
in \lambdaio{} are limited, we still tested a number of things. These are
some agents we have tried that work as intended:
\begin{lstlisting}
(* bad: this agent does nothing *)
let lazy_agent : exp = ELam (ELam EUnit)

(* good: this agent writes the expected string on the file *)
let write_agent : exp =
  ELam (ELam (ECase (EOpen (EVar 1)) (EApp (ELam (EClose (EVar 1))) (EWrite (EVar 0) (EVar 1))) EUnit ))

(* bad: this agent writes the string to the file twice *)
let write_twice_agent : exp =
  ELam (ELam (ECase (EOpen (EVar 1))
                    (EApp (ELam (EApp (ELam (EClose (EVar 2)))
                          (EWrite (EVar 1) (EVar 2))))
                          (EWrite (EVar 0) (EVar 1))) EUnit))
\end{lstlisting}

\subsection{Unverified Compilation Step from \lambdaio{} to \lambdabox{}}
\label{sec:further-compilation}

\begin{figure}[h]
\centering
\begin{tikzpicture}[
  auto,
  node distance=2.4cm
]

\tikzstyle{block} = [
  draw,
  rectangle,
  minimum height=2.6em,
  minimum width=4.5em,
  align=center
]

\node[block, rounded corners] (fst)
  {\small\textit{\lambdaio{}}\\[-2pt]\small\texttt{.fst}};

\node[block, rounded corners, right of=fst] (ast)
  {\small\textit{\lambdabox{}}\\[-2pt]\small\texttt{.ast}};

\node[block, rounded corners, right of=ast] (mlf)
  {\small\textit{Malfunction}\\[-2pt]\small\texttt{.mlf}};

\node[block, rounded corners, right of=mlf] (cmx)
  {\small\textit{OCaml obj.}\\[-2pt]\small\texttt{.cmx}};

\node[block, rounded corners, right of=cmx] (exe)
  {\small\textit{Executable}};

\node[block, rounded corners, below of=cmx, node distance=1.4cm] (axioms)
  {\small\textit{IO runtime}\\[-2pt]\small\texttt{axioms.cmx}};

\node[block, rounded corners, below of=mlf, node distance=1.4cm] (axiomsml)
  {\small\textit{IO runtime}\\[-2pt]\small\texttt{axioms.ml}};

\draw[->] (fst) -- (ast);
\draw[->] (ast) -- (mlf);
\draw[->] (mlf) -- (cmx);
\draw[->] (cmx) -- (exe);
\draw[->] (axiomsml) -- (axioms);
\draw[->] (axioms) -- (exe);
\end{tikzpicture}
\vspace{0.5em}
\caption{Compilation pipeline from \lambdaio{} to executable.}
\label{fig:exec-pipeline}
\end{figure}

To build an actual executable from the programs we can write
in \seiostar{}, we complement the proposed extraction mechanism with
an unverified step compiling from \lambdaio{} to \lambdabox{}
(\emph{lambdabox}), and then using the tooling provided by the
Peregrine project~\cite{Peregrine} to compile to OCaml
(Malfunction~\cite{coq-to-ocaml,dolan2016malfunctional}) to obtain code that can be later
linked with a runtime in OCaml.

The \lambdabox{} language was introduced in the context of program
extraction for Rocq~\cite{Letouzey04:thesis}. It is a variant
of untyped $\lambda$-calculus with a construction $\square$
(box) that replaces proofs and types that are computationally
irrelevant. In our use case, we do not erase proofs, as \lambdaio{}
already does not contain any, and we use it just as an intermediate language
that allows us to use the rest of the compilation pipeline provided by Peregrine.

Using Peregrine, the \lambdabox{} code is translated to Malfunction,
which is a low-level untyped program representation that can be used
to compile to OCaml objects. The Peregrine framework uses this program
representation as its output for OCaml. This output can later be
compiled using Malfunction tooling~\cite{dolan2016malfunctional}\ch{no
clue what's binary this is this talking about}\er{replaced with
tooling and give ref.}  to generate OCaml objects, which can be linked
with other OCaml objects that implement IO primitives such
as \ls{openfile}, \ls{readfile}, etc. as well as other primitives
of \lambdabox{} such as string equality. This compilation and linking
sequence is presented in \autoref{fig:exec-pipeline}.\tw{repetition}

\section{Related Work}
\label{sec:related}

\ca{TODO: Intro paragraph. There is something in the comments}



\ca{
  We should discuss maybe again about what this means
  and exactly what is the plan to represent OCaml
  computations in \fstar{}.
  Maybe Remy Seassau work is relevant to us:
  \url{https://remyjck.github.io/papers/osiris.pdf}.
}

\paragraph{\bf{Certifying Extraction}}
Work on {\em CakeML synthesis}~\cite{cakeml-pure-synthesis,cakeml-effectful-synthesis}
shows how starting from a shallowly embedded program in HOL one
can synthesize an equivalent CakeML~\cite{cakeml} program (a subset of Standard ML).
Their synthesis happens outside of HOL, in Standard ML, and it
produces a deeply embedded program in HOL and a correctness proof
that has to be accepted by HOL, which is a form of translation valuation.
They show how to synthesize CakeML programs from HOL functions
that feature polymorphism, first-order references, arrays, exceptions, and IO
operations.\ch{Do they also use monads for the effects?}
We think we can extend \seiostar{} with these features by following
their lead to obtain certifying extraction from \fstar{}
with a smaller TCB\ch{why smaller? shouldn't forget that \fstar{} has a huge TCB and HOL a tiny one}
and stronger security guarantees
than in their work~\cite{cakeml-pure-synthesis,cakeml-effectful-synthesis}.
Another difference comes from the languages we use,
they work in HOL, a simply typed language,
while we work in \fstar{}, a dependently typed language,
which enabled us to factor out relational quotation as a separate step.

{\em Œuf~\cite{oeuf}} is a compiler with end-to-end guarantees
from a simply typed subset of Gallina (the language of Rocq) to Assembly code.
They use translation validation for quotation, which
produces a deeply embedded program.
They validate this
by showing in Rocq that the denotation of the deep embedding is equal to the Gallina program.
They define denotation as a pure Rocq function that takes a deeply embedded program and produces
a shallowly embedded one.
This is similar to what happens in relational quotation, where the typing derivation also
constructs a program that is checked to be equal with the program that is quoted.
From our experience, using a typing relation works better than
separately defining a deep embedding and a denotation function.
For example, it is more natural to keep in the
typing derivation extra information necessary to
reconstruct the shallow embedding,
information which may have no place in a standard deep embedding.
We used relational quotation for a language with the IO effect,
while Œuf has no support for effects.
\ca{what about their shims? They do not seem so interesting}

{\em \citet{certextraction-coq}} present a certifying extraction
from pure Rocq functions of simple polymorphic types to an
untyped call-by-value \(\lambda\)-calculus.
They present a plugin for Rocq (a metaprogram) that given a shallowly embedded
program, produces a deeply embedded program.
The plugin comes with a series of tactics that automate the
process of validating that the two programs are related.
They have no support for effects, while \seiostar{} supports IO.

{\em \citet{rupicola}} present a Rocq framework
that treats compilation as a proof-search problem---\IE
finding a deeply embedded low-level program that is in a relation with the
shallowly embedded program.
This is also a form of translation validation, since the search is done using
tactics (metaprograms) whose result has to be type-checked by Rocq.
The framework may not know how to automatically compile a program, in which case it
shows where it got stuck, and the proof engineer can take control.
They illustrate the versatility of the framework by showing how to 
compile Gallina programs using
terminating loops,
various flat data structures such as mutable first-order cells and arrays,
nondeterminism, and IO.
\citet{rupicola} present their framework as
translation validation tool for specific performance-critical programs,
directly targetting a low-level program.\ch{May need to clarify what this last part means.}
We on the other hand try to minimize 
the reliance on translation validation.\ch{
    only difference?\ca{this is the biggest difference}}\ch{Do they do secure compilation?}

The work above and our relational quotation {\em share a challenge}
that comes from trying to relate a shallowly embedded program
with a deeply embedded program:
one is able to
relate only what is expressible by a user defined relation
(\EG pattern matching is hard to do: when synthesizing CakeML~\cite{cakeml-effectful-synthesis,cakeml-pure-synthesis},
    they first automatically generate the relation based on the datatypes used).
\newremove{
We came up with the {\em typing relation for shallow embeddings} after reading work
on certifying extraction,
but, to our knowledge, it does not exist per se in literature.
}

\newtext{
  \paragraph{\bf{Typing Relation for Shallow Embeddings}}
  {\em \citet{mcbride-coincidences}} presents the first
  deep embedding indexed by shallow types, and claims it
  is a ``new reflective technology''.
  McBride's idea was advanced 
  by~{\em\citet{deeper-shallow-embeddings}},
  who show how to do a ``deeper shallow embedding''
  of a dependently-typed language by using an embedding
  very similar to our typing relation for shallow embeddings.
  They show an embedding that supports dependent types and/or affine types,
  something we did not try. We show, however, how to support monadic code and refinement types.
  To our knowledge, we are the first to use such a typing relation for
  doing quotation and certifying extraction.

  {\em \citet{embedding-by-unembedding}} propose a framework that leverages unembedding,
  a reification technique that converts
  finally-tagless (shallow) representations to de Bruijn-indexed (deep) terms, with
  formal correctness guarantees.
  Unembedding does not apply to our monadic shallowly
  embedded code because it is not in finally-tagless style.

  Some similarities exist also with work on canonical structures~\cite{gonthier2013make}.

}

\paragraph{\bf{Logical Relations}}
Our logical relations,
and their related proofs,
follow the logical relation for semantic typing presented by \citet{ahmed2004semantics} (\S
2.2) for a \(\lambda\)-calculus extended with dynamically allocated immutable cells.
We differ from it in three ways:
our logical relations are cross-language~\cite{NeisHKMDV15, Ahmed15}
  and asymmetric, connecting a shallow and a deep embedding.
They relate a program of a FGCBV language to one of a CBV language, and 
relates two trace-producing semantics.\ca{we also differ because we do IO, but
  somehow it does not seem very relevant. maybe it is covered by the 
  trace-producing semantics?}\ch{The ``trace-producing semantics'' is mentioned 100 times,
  but it's relevance {\bf never} explained. To me it is the most standard way to give an operational semantics for IO.
    Why are we making such a fuss about something standard?
    Here in related work (comparison with CakeML?) is the place where it should be explained, and if it's standard
    then also only place where it needs to be mentioned.}

Asymmetric relations are common in work on certified extraction,
and \citet{certextraction-coq} and
work on CakeML synthesis~\cite{cakeml-effectful-synthesis,cakeml-pure-synthesis}
also use logical relations.
The most complex relation is used in CakeML synthesis:
it supports multiple effects (first-order state, exceptions and IO), polymorphic functions and
total recursive functions,
and relates evaluation semantics.\ch{unclear to me what distinction you're trying to
  make here between evaluation and trace-producing semantics}
Our logical relations are different because
they relate trace-producing semantics\ch{unclear why this is not the same in CakeML, which also does IO}
and it implies the strongest secure compilation criterion (RrHP) of~\citet{AbateBGHPT19}.
  
Logical relations were used to prove fully abstract compilation~\cite{NewBA16,
  DevriesePP16, MarcosSurvey} and \citet{AbateBGHPT19} later show that they can also
be used to prove their stronger RrHP secure compilation criterion for a simple
translation from a statically typed language to a dynamically typed one.

\paragraph{\bf{Verified Extraction}}
CakeML~\cite{cakeml} is a subset
of Standard ML that has a verified compiler that can
target ARMv6, ARMv8, x86-x64, MIPS-64, RISC-V and Silver RSA
architectures.\ch{What's the connection to {\em extraction} (title of this heading)?
  This is just a traditional verified compiler, like CompCert.
  Do I need to read the {\em next} paragraph to understand why you're telling me all this?
  How about you start with the interesting part?}
Together with their synthesis from HOL to CakeML, they have 
end-to-end guarantees of correctness from the original program
written in HOL to the low-level code that targets a specific
architecture.
It would be interesting to extend \seiostar{} to target
CakeML with a verified secure compilation step.

\citet{Isabelle2CakeML} and \citet{verified_dafny_extraction} present verified extraction from Isabelle/HOL and Dafny to
CakeML, but there is no formal connection between the
formalized meta-theory of Isabelle/HOL or Dafny and the languages themselves.\ch{No
  clue what the part after ``but'' is trying to say. What languages and connection are you talking about?}

\citet{coq-to-ocaml} present a verified compiler
from Gallina to Malfunction (an intermediate untyped language of OCaml).
Their compiler is based on the existing MetaRocq~\cite{metacoq2} project
that formalizes Rocq in itself.
%
%
They prove a realizability relation that guarantees that the
extracted untyped code operationally behaves
as the initial Rocq program.
They do not support
effectful code, 
and quotation is neither verified nor certifying.

CertiCoq~\cite{certicoq17} and ConCert~\cite{concert-rocq, certicoq-wasm}
are two partially verified compilers for Rocq that target C, Elm, Rust, WebAssembly and
blockchain languages. They also rely on unverified quotation.

\paragraph{\bf{Secure Compilation of Verified Code}}

There are two secure compilation frameworks for \fstar{},
\sciostar{}~\cite{sciostar} for verified IO programs, 
and \secrefstar{}~\cite{secrefstar} for verified stateful programs.
\sciostar{} and \secrefstar{} look at how to
convert refinement types, pre- and post-conditions into dynamic
checks using higher-order contracts and a reference monitor.
\seiostar{} is orthogonal to them, since it focuses on
providing guarantees for the extraction step itself.
They are both proved to satisfy RrHP, but they are not fully verified:
they rely on the unverified extraction mechanism of \fstar{} to OCaml.
\newtext{Moreover, even if we prove the same criterion, the proofs are very different.
Their proof~\cite{sciostar,secrefstar} of RrHP follows directly
from a syntactic inversion law, while our \seiostar{} proof of RrHP is
semantic and involves two logical relations.
This difference is necessary, because while \sciostar{}/\secrefstar{} only
compiles between shallow embeddings, our \seiostar{} work extracts a shallow
embedding to a deeply embedded language with a standard operational semantics.}
Beyond such secure extraction from proof assistants, there
are many other applications of secure compilation~\cite{MarcosSurvey, AbateBGHPT19}.

\section{Conclusion and Future Work}
\label{sec:future}

We have shown that formally secure extraction of shallowly embedded effectful
programs can be achieved by using translation validation only for a first
step---relational quotation---while using a once-and-for-all verified syntax
generation function for the rest. We have illustrated this general idea for \fstar{} by
developing \seiostar{}, a framework that extracts \fstar{} programs with IO
while providing machine-checked secure compilation guarantees (RrHP).

\newtext{
\seiostar{} supports extraction of programs with file-based IO,
whose types are simple types, possibly decorated with refinements.
In the future, it will be interesting to address some of its limitations,
and add support for recursive functions and dependent types.
To quote recursive \fstar{} functions we would need
\fstar{} to expose the combinators used to define recursive functions.
For now, we defined an iterator for the specific inductive type of natural numbers
and showed how to extract functions written in terms of that iterator~\cite{oeuf}.
We think this method could scale to a user defined fixpoint
combinator~\cite{cakeml-effectful-synthesis,cakeml-pure-synthesis,rupicola},
but the combinator would still not be idiomatic and it would be
difficult to use compared to \fstar{}'s native support for recursion.
To support dependent types, we would have to modify
the typing relation to be closer to \citet{deeper-shallow-embeddings}
by having the \fstar{} types explicitly depend on the evaluation environment, just like we did for values.
We expect that quoting dependently-typed programs would
still follow the syntactic structure of the program and the
type checking would be similar to how it works for simply typed programs.
Redoing the proofs of security would require to also extend the logical relations.
Logical relations for dependent types tend to be much more involved and rely on strong meta-theoretical assumptions such as induction-recursion, ~\cite{abel2013normalization} on impredicativity~\cite{adjedj2024martin,jang2025mctt,liu2024functional}, or alternatively by considering only a finite number of universes~\cite{abel2017decidability,poiret2026divide}.
While this seems mostly orthogonal from the focus of our work, the intricacies
of logical relations prevent us from making a claim about the difficulty or
feasibility of the extension.}

Improving the performance of the relational quotation metaprogram for refinement types
is necessary before doing other extensions.
We show how to support refinement types by defining
a typing relation that computes a precondition and has one 
typing rule specific to refinements.
Having the precondition as an argument on each typing rule consumes a lot of memory
because the precondition grows with the size of the derivation.
The metaprogram also overuses the subtyping rule (\ls{QRef}),
by using it in all position where typing information comes from the environment,
but not all positions require subtyping.

In the future, we would like to connect \seiostar{} with \sciostar{}
~\cite{sciostar}, the previous secure compilation framework for verified
\fstar{} programs with IO, already discussed at the end of \autoref{sec:related} above.
This raises several challenges though:
\newremove{First, while \sciostar{} turns some of the refinement types and
pre-/post-conditions on the interface between program and context into dynamic
checks, the compiled programs produced by \sciostar{} can still make use of
these features (\EG internally).
We believe relational quotation can be extended to these features, but this is
made challenging by the fact that in \fstar{} logical facts from the context
(\EG refinements on arguments or function pre-conditions) can be used
automatically by the SMT solver, without any syntactic marker in the program.}%
First, \sciostar{} works around \fstar's lack of general effect polymorphism by
making use of a restricted form called flag polymorphism to model that the
context is prevented from directly calling the IO primitives to avoid the access
control done by a reference monitor.
Second, \sciostar{} makes use of \fstar's effect mechanism, and while effects in
\fstar{} can be reified to their underlying monadic
representation~\cite{relational}, support for reasoning about the obtained
monadic computations is at the moment still too limited to be able to do the
extrinsic proofs that would be needed to connect \sciostar{} and \seiostar{}.

It is worth noting that these two challenges are specific to \fstar{} and the
general idea of using relational quotation to split extraction into two distinct
phases could also be applied to more standard dependently typed proof assistants
like Rocq and Lean.
For instance, it would be interesting to combine and extend the ideas of both
\sciostar{} and \seiostar{} and apply them to Loom~\cite{GladshteinPZKS26}, a
recent verification framework for effectful Lean programs, which like \fstar{}
is based on Dijkstra monads~\cite{dm4all}.
This would also require extending relational quotation to more effects, but prior
work on CakeML synthesis~\cite{cakeml-pure-synthesis,
  cakeml-effectful-synthesis} and relational compilation~\cite{rupicola} suggest
this should be possible.


\tbd{
Certifying extraction can be used to achieve end-to-end verified extraction.
One example of that is the CakeML~\cite{cakeml} project:
the compiler for CakeML is verified in HOL (from parsing to pretty-printing), and certifying extraction ~\cite{cakeml-pure-synthesis,cakeml-effectful-synthesis}
(what they call synthesis)
is used to extract the verified compiler. The compiler then provides guarantees
not using translation validation, but by being end-to-end verified,
which we know it satisfies because of the certifying process.
Then to achieve end-to-end verified extraction, one has to 
formalize the meta-theory necessary to verify extraction,
and to scale 
certifying extraction to the subset used to formalize the meta-theory.}

\ca{Formalizing the meta-theory, scaling relation quotation to the features used,
  relate the 
  }

\ca{connect to lambda box? we already said in Related work>verified extraction about targeting an existing
  verified compilation chain}

\ca{create framework}

\ca{
  Can one certify extraction of MetaRocq? Challenges:
  \begin{itemize}

  \item Extend relational quotation to support quotation of MetaRocq.
     One needs dependent-types, what else?
     This may be easier than it sounds because we target just the Rocq subset in which
     MetaRocq is written into (probably a fixed number of universes?)
     and probably a fixed number of types and usages of pattern matching.
     We do not want to quote the proofs either, just the computational parts.
  \item We only have to certify the typing derivations
    (which can be manually written and proven if automation proves difficult).
    Initially, the certification would be done by the OCaml implementation of Rocq,
    but then one could re-certify with the verified extracted version (probably to slow \citet{coq-to-ocaml}).
  \item What to target? Can we reuse 
       verified extraction to Malfunction \citet{coq-to-ocaml}?
       Could we target PCUIC? Or lambda box after erasure?
  \item Establish a logical relation
     between the formalization and the extracted code
     that gives us correctness guarantees.
  \item Would the final statement be in terms of the logical relation?
  \end{itemize}
}

\ifanon\else
\begin{acks}
We are grateful to Guido Mart\'{i}nez, Yong Kiam Tan, and Jérémy Thibault for the insightful discussions.
We also thank the anonymous reviewers of PriSC'26 and ICFP'26 for their helpful feedback.
This work was in part funded by the \grantsponsor{1}{European Research Council}{https://erc.europa.eu/}
under \ifcamera\else ERC\fi{} Starting Grant SECOMP (\grantnum{1}{715753}),
by the German Federal Ministry of Education and Research BMBF (grant 16KISK038, project 6GEM),
and by the Deutsche Forschungsgemeinschaft (DFG, German Research Foundation)
under Germany's Excellence Strategy -- EXC 2092 CASA -- 390781972.
Exequiel Rivas was funded by the Estonian Research Council starting grant PSG749.
Danel Ahman was funded by the Estonian Research Council grant PRG2764.
\end{acks}
\fi


\ifcamera
\bibliographystyle{ACM-Reference-Format}
\else
\bibliographystyle{abbrvnaturl}
\fi

\bibliography{fstar}

\begin{thebibliography}{53}
\providecommand{\natexlab}[1]{#1}

\bibitem[Abate et~al.(2019)Abate, Blanco, Garg, Hri\cb{t}cu, Patrignani, and
  Thibault]{AbateBGHPT19}
C.~Abate, R.~Blanco, D.~Garg, C.~Hri\cb{t}cu, M.~Patrignani, and J.~Thibault.
\newblock \href {http://dx.doi.org/10.1109/CSF.2019.00025} {Journey beyond full
  abstraction: Exploring robust property preservation for secure compilation}.
\newblock \iflongrefs{In \emph{32nd {IEEE} Computer Security Foundations
  Symposium ({CSF})}}\else{\emph{CSF}}\fi{}. 2019.

\bibitem[Abel(2013)]{abel2013normalization}
A.~Abel.
\newblock \href {https://www.cse.chalmers.se/~abela/habil.pdf}
  {\emph{Normalization by evaluation: Dependent types and impredicativity}}.
\newblock Habilitation thesis, Ludwig-Maximilians-Universit{\"a}t M{\"u}nchen,
  2013.

\bibitem[Abel et~al.(2017)Abel, {\"O}hman, and Vezzosi]{abel2017decidability}
A.~Abel, J.~{\"O}hman, and A.~Vezzosi.
\newblock \href {http://dx.doi.org/10.1145/3158111} {Decidability of conversion
  for type theory in type theory}.
\newblock \emph{Proceedings of the ACM on Programming Languages}, 2\penalty0
  (POPL):\penalty0 1--29, 2017.

\bibitem[Abrahamsson et~al.(2020)Abrahamsson, Ho, Kanabar, Kumar, Myreen,
  Norrish, and Tan]{cakeml-effectful-synthesis}
O.~Abrahamsson, S.~Ho, H.~Kanabar, R.~Kumar, M.~O. Myreen, M.~Norrish, and
  Y.~K. Tan.
\newblock \href {http://dx.doi.org/10.1007/s10817-020-09559-8} {Proof-producing
  synthesis of {CakeML} from monadic {HOL} functions}.
\newblock \emph{Journal of Automated Reasoning ({JAR})}, 2020.

\bibitem[Adjedj et~al.(2024)Adjedj, Lennon-Bertrand, Maillard, P{\'e}drot, and
  Pujet]{adjedj2024martin}
A.~Adjedj, M.~Lennon-Bertrand, K.~Maillard, P.-M. P{\'e}drot, and L.~Pujet.
\newblock \href {http://dx.doi.org/10.1145/3636501.3636951} {Martin-l{\"o}f
  {\`a} la coq}.
\newblock In \emph{Proceedings of the 13th ACM SIGPLAN International Conference
  on Certified Programs and Proofs}, 2024.

\bibitem[Ahman and Uustalu(2013)]{ahman13update}
D.~Ahman and T.~Uustalu.
\newblock \href {http://dx.doi.org/10.4230/LIPIcs.TYPES.2013.1} {Update monads:
  Cointerpreting directed containers}.
\newblock \iflongrefs{In \emph{19th International Conference on Types for
  Proofs and Programs, {TYPES} 2013, April 22-26, 2013, Toulouse,
  France}}\else{\emph{TYPES}}\fi{}, 2013.

\bibitem[Ahman et~al.(2017)Ahman, Hri\cb{t}cu, Maillard, Mart\'inez, Plotkin,
  Protzenko, Rastogi, and Swamy]{dm4free}
D.~Ahman, C.~Hri\cb{t}cu, K.~Maillard, G.~Mart\'inez, G.~Plotkin, J.~Protzenko,
  A.~Rastogi, and N.~Swamy.
\newblock \href {http://dx.doi.org/10.1145/3009837.3009878} {Dijkstra monads
  for free}.
\newblock \iflongrefs{In \emph{44th ACM SIGPLAN Symposium on Principles of
  Programming Languages (POPL)}}\else{\emph{POPL}}\fi{}. \iflongrefs{Jan.
  }\fi{}2017.

\bibitem[Ahman et~al.(2018)Ahman, Fournet, Hri\cb{t}cu, Maillard, Rastogi, and
  Swamy]{preorders}
D.~Ahman, C.~Fournet, C.~Hri\cb{t}cu, K.~Maillard, A.~Rastogi, and N.~Swamy.
\newblock \href {http://dx.doi.org/10.1145/3158153} {Recalling a witness:
  Foundations and applications of monotonic state}.
\newblock \emph{{PACMPL}}, 2\penalty0 ({POPL}):\penalty0 65:1--65:30,
  \iflongrefs{jan }\fi{}2018.

\bibitem[Ahman et~al.(2026)Ahman, Bhargavan, Bond, Bosamiya, Brzuska,
  Delignat-Lavaud, Fournet, Fromherz, Gibson, Hawblitzel, Hri{\cb{t}}cu,
  Kohlweiss, Mart\'{\i}nez, Ni, Parno, Protzenko, Ramananandro, Rastogi, Rivas,
  Swamy, and Zanella-B\'{e}guelin]{everest-toplas26}
D.~Ahman, K.~Bhargavan, B.~Bond, J.~Bosamiya, C.~Brzuska, A.~Delignat-Lavaud,
  C.~Fournet, A.~Fromherz, S.~Gibson, C.~Hawblitzel, C.~Hri{\cb{t}}cu,
  M.~Kohlweiss, G.~Mart\'{\i}nez, H.~Ni, B.~Parno, J.~Protzenko,
  T.~Ramananandro, A.~Rastogi, E.~Rivas, N.~Swamy, and S.~Zanella-B\'{e}guelin.
\newblock \href {http://dx.doi.org/10.1145/3805702} {Project {Everest}:
  Perspectives from developing industrial-grade high-assurance software}.
\newblock \emph{ACM Trans. Program. Lang. Syst.}, 48\penalty0 (2),
  \iflongrefs{June }\fi{}2026.

\bibitem[Ahmed(2015)]{Ahmed15}
A.~Ahmed.
\newblock \href {http://dx.doi.org/10.4230/LIPIcs.SNAPL.2015.15} {Verified
  compilers for a multi-language world}.
\newblock \iflongrefs{In \emph{1st Summit on Advances in Programming
  Languages}}\else{\emph{SNAPL}}\fi{}. 2015.

\bibitem[Ahmed(2004)]{ahmed2004semantics}
A.~J. Ahmed.
\newblock \href {http://www.ccs.neu.edu/home/amal/ahmedthesis.pdf}
  {\emph{Semantics of Types for Mutable State}}.
\newblock PhD thesis, Princeton University, 2004.

\bibitem[Anand et~al.(2017)Anand, Appel, Morrisett, Paraskevopoulou, Pollack,
  Belanger, Sozeau, and Weaver]{certicoq17}
A.~Anand, A.~Appel, G.~Morrisett, Z.~Paraskevopoulou, R.~Pollack, O.~S.
  Belanger, M.~Sozeau, and M.~Weaver.
\newblock \href
  {https://popl17.sigplan.org/details/main/9/CertiCoq-A-verified-compiler-for-Coq}
  {{CertiCoq}: A verified compiler for {Coq}}.
\newblock In \emph{3rd Workshop on Coq for Programming Languages ({CoqPL})},
  2017.

\bibitem[Andrici et~al.(2024)Andrici, Ciob{\^{a}}c{\u{a}}, Hri{\cb{t}}cu,
  Mart{\'{\i}}nez, Rivas, Tanter, and Winterhalter]{sciostar}
C.-C. Andrici, {\c{S}}.~Ciob{\^{a}}c{\u{a}}, C.~Hri{\cb{t}}cu,
  G.~Mart{\'{\i}}nez, E.~Rivas, {\'{E}}.~Tanter, and T.~Winterhalter.
\newblock \href {http://dx.doi.org/10.1145/3632916} {Securing verified {IO}
  programs against unverified code in {F}$^{\star}$}.
\newblock \emph{Proc. {ACM} Program. Lang.}, 8\penalty0 ({POPL}):\penalty0
  2226--2259, 2024.

\bibitem[Andrici et~al.(2025)Andrici, Ahman, Hri\cb{t}cu, Icleanu,
  Mart\'{\i}nez, Rivas, and Winterhalter]{secrefstar}
C.-C. Andrici, D.~Ahman, C.~Hri\cb{t}cu, R.~Icleanu, G.~Mart\'{\i}nez,
  E.~Rivas, and T.~Winterhalter.
\newblock \href {http://dx.doi.org/10.1145/3747522} {{SecRef}$^{\star}$:
  Securely sharing mutable references between verified and unverified code in
  {F}$^{\star}$}.
\newblock \emph{Proc. ACM Program. Lang.}, 9\penalty0 (ICFP), \iflongrefs{Aug.
  }\fi{}2025.

\bibitem[Andrici et~al.(2026)Andrici, Pribisova, Ahman, Hri\cb{t}cu, Rivas, and
  Winterhalter]{andrici_2026_20534723}
C.-C. Andrici, A.~Pribisova, D.~Ahman, C.~Hri\cb{t}cu, E.~Rivas, and
  T.~Winterhalter.
\newblock \href {http://dx.doi.org/10.5281/zenodo.20534723} {Misquoted no more:
  Securely extracting f* programs with io - artifact}, \iflongrefs{June
  }\fi{}2026.

\bibitem[Annenkov et~al.(2022)Annenkov, Milo, Nielsen, and
  Spitters]{concert-rocq}
D.~Annenkov, M.~Milo, J.~B. Nielsen, and B.~Spitters.
\newblock \href {http://dx.doi.org/10.1017/S0956796822000077} {Extracting
  functional programs from {Coq}, in {Coq}}.
\newblock \emph{Journal of Functional Programming}, 32:\penalty0 e11, 2022.

\bibitem[Carneiro(2024)]{DBLP:journals/corr/abs-2403-14064}
M.~Carneiro.
\newblock \href {http://dx.doi.org/10.48550/ARXIV.2403.14064} {{Lean4Lean}:
  Towards a formalized metatheory for the {Lean} theorem prover}.
\newblock \emph{CoRR}, abs/2403.14064, 2024.

\bibitem[Chapman et~al.(2026)Chapman, Dima, Escot, Forster, Melkonian, Nielsen,
  Sozeau, and Spitters]{Peregrine}
J.~Chapman, S.~Dima, L.~Escot, Y.~Forster, O.~Melkonian, E.~Nielsen, M.~Sozeau,
  and B.~Spitters.
\newblock \href {https://peregrine-project.github.io} {Peregrine project: a
  unified middle-end for code generation from proof assistants}, 2026.

\bibitem[Cohen and Johnson-Freyd(2024)]{10.1145/3632902}
J.~M. Cohen and P.~Johnson-Freyd.
\newblock \href {http://dx.doi.org/10.1145/3632902} {A formalization of core
  {Why3} in {Coq}}.
\newblock \emph{Proc. ACM Program. Lang.}, 8\penalty0 (POPL), \iflongrefs{Jan.
  }\fi{}2024.

\bibitem[Devriese et~al.(2016)Devriese, Patrignani, and Piessens]{DevriesePP16}
D.~Devriese, M.~Patrignani, and F.~Piessens.
\newblock \href {http://dx.doi.org/10.1145/2914770.2837618} {Fully-abstract
  compilation by approximate back-translation}.
\newblock \iflongrefs{In \emph{43nd Annual {ACM} {SIGPLAN-SIGACT} Symposium on
  Principles of Programming Languages}}\else{\emph{POPL}}\fi{}, 2016.

\bibitem[Dolan(2016)]{dolan2016malfunctional}
S.~Dolan.
\newblock \href {https://stedolan.net/talks/2016/malfunction/malfunction.pdf}
  {Malfunctional programming}.
\newblock In \emph{ML Family Workshop 2016}, 2016.

\bibitem[Forster and Kunze(2019)]{certextraction-coq}
Y.~Forster and F.~Kunze.
\newblock \href {http://dx.doi.org/10.4230/LIPICS.ITP.2019.17} {A certifying
  extraction with time bounds from {Coq} to call-by-value lambda calculus}.
\newblock In J.~Harrison, J.~O'Leary, and A.~Tolmach, editors, \emph{10th
  International Conference on Interactive Theorem Proving, {ITP} 2019,
  Portland, OR, USA, September 9-12, 2019}. 2019.

\bibitem[Forster et~al.(2024)Forster, Sozeau, and Tabareau]{coq-to-ocaml}
Y.~Forster, M.~Sozeau, and N.~Tabareau.
\newblock \href {http://dx.doi.org/10.1145/3656379} {Verified extraction from
  {Coq} to {OCaml}}.
\newblock \emph{Proc. {ACM} Program. Lang.}, 8\penalty0 ({PLDI}):\penalty0
  52--75, 2024.

\bibitem[Gladshtein et~al.(2026)Gladshtein, P{\^{\i}}rlea, Zhao, Kurin, and
  Sergey]{GladshteinPZKS26}
V.~Gladshtein, G.~P{\^{\i}}rlea, Q.~Zhao, V.~Kurin, and I.~Sergey.
\newblock \href {http://dx.doi.org/10.1145/3776719} {Foundational multi-modal
  program verifiers}.
\newblock \emph{Proc. {ACM} Program. Lang.}, 10\penalty0 ({POPL}):\penalty0
  2233--2264, 2026.

\bibitem[Gonthier et~al.(2013)Gonthier, Ziliani, Nanevski, and
  Dreyer]{gonthier2013make}
G.~Gonthier, B.~Ziliani, A.~Nanevski, and D.~Dreyer.
\newblock \href {http://dx.doi.org/10.1145/2034773.2034798} {How to make ad hoc
  proof automation less ad hoc}.
\newblock \emph{Journal of Functional Programming}, 23\penalty0 (4):\penalty0
  357--401, 2013.

\bibitem[Grimm et~al.(2018)Grimm, Maillard, Fournet, Hri\cb{t}cu, Maffei,
  Protzenko, Ramananandro, Rastogi, Swamy, and
  {Zanella-B\'eguelin}]{relational}
N.~Grimm, K.~Maillard, C.~Fournet, C.~Hri\cb{t}cu, M.~Maffei, J.~Protzenko,
  T.~Ramananandro, A.~Rastogi, N.~Swamy, and S.~{Zanella-B\'eguelin}.
\newblock \href {http://dx.doi.org/10.1145/3167090} {A monadic framework for
  relational verification: Applied to information security, program
  equivalence, and optimizations}.
\newblock \iflongrefs{In \emph{The 7th ACM SIGPLAN International Conference on
  Certified Programs and Proofs}}\else{\emph{CPP}}\fi{}, \iflongrefs{Jan.
  }\fi{}2018.

\bibitem[Hupel and Nipkow(2018)]{Isabelle2CakeML}
L.~Hupel and T.~Nipkow.
\newblock \href {http://dx.doi.org/10.1007/978-3-319-89884-1\_35} {A verified
  compiler from {Isabelle/HOL} to {CakeML}}.
\newblock In A.~Ahmed, editor, \emph{Programming Languages and Systems - 27th
  European Symposium on Programming, {ESOP} 2018, Held as Part of the European
  Joint Conferences on Theory and Practice of Software, {ETAPS} 2018,
  Thessaloniki, Greece, April 14-20, 2018, Proceedings}. 2018.

\bibitem[Jang et~al.(2025)Jang, Gaulin, Hu, and Pientka]{jang2025mctt}
J.~Jang, A.~Gaulin, J.~Z. Hu, and B.~Pientka.
\newblock \href {http://dx.doi.org/10.1145/3747511} {Mctt: A verified kernel
  for a proof assistant}.
\newblock \emph{Proceedings of the ACM on Programming Languages}, 9\penalty0
  (ICFP):\penalty0 190--221, 2025.

\bibitem[Kumar et~al.(2014)Kumar, Myreen, Norrish, and Owens]{cakeml}
R.~Kumar, M.~O. Myreen, M.~Norrish, and S.~Owens.
\newblock \href {http://dx.doi.org/10.1145/2535838.2535841} {{CakeML}: a
  verified implementation of {ML}}.
\newblock In \emph{The 41st Annual {ACM} {SIGPLAN-SIGACT} Symposium on
  Principles of Programming Languages, {POPL}}. 2014.

\bibitem[Letouzey(2004)]{Letouzey04:thesis}
P.~Letouzey.
\newblock \href {https://tel.archives-ouvertes.fr/tel-00150912}
  {\emph{Programmation fonctionnelle certifi{\'{e}}e : L'extraction de
  programmes dans l'assistant Coq. (Certified functional programming: Program
  extraction within Coq proof assistant)}}.
\newblock PhD thesis, University of Paris-Sud, Orsay, France, 2004.

\bibitem[Levy et~al.(2003)Levy, Power, and Thielecke]{fgcbv}
P.~B. Levy, J.~Power, and H.~Thielecke.
\newblock \href {http://dx.doi.org/10.1016/S0890-5401(03)00088-9} {Modelling
  environments in call-by-value programming languages}.
\newblock \emph{Inf. Comput.}, 185\penalty0 (2):\penalty0 182--210, 2003.

\bibitem[Liesnikov and Cockx(2024)]{DBLP:conf/aplas/LiesnikovC24}
B.~Liesnikov and J.~Cockx.
\newblock \href {http://dx.doi.org/10.1007/978-981-97-8943-6\_4} {Building a
  correct-by-construction type checker for a dependently typed core language}.
\newblock In O.~Kiselyov, editor, \emph{Programming Languages and Systems -
  22nd Asian Symposium, {APLAS} 2024, Kyoto, Japan, October 22-24, 2024,
  Proceedings}. 2024.

\bibitem[Liu et~al.(2024)Liu, Chan, and Weirich]{liu2024functional}
Y.~Liu, J.~Chan, and S.~Weirich.
\newblock \href {https://electriclam.com/papers/mltt.pdf} {Functional pearl:
  Short and mechanized logical relation for dependent type theories}.
\newblock 2024.

\bibitem[Maillard et~al.(2019)Maillard, Ahman, Atkey, Mart\'{\i}nez,
  Hri\cb{t}cu, Rivas, and Tanter]{dm4all}
K.~Maillard, D.~Ahman, R.~Atkey, G.~Mart\'{\i}nez, C.~Hri\cb{t}cu, E.~Rivas,
  and E.~Tanter.
\newblock \href {http://dx.doi.org/10.1145/3341708} {Dijkstra monads for all}.
\newblock \iflongrefs{\emph{Proc. ACM Program.
  Lang.}}\else{\emph{PACMPL}}\fi{}, 3\penalty0 (ICFP), \iflongrefs{July
  }\fi{}2019.

\bibitem[Mart\'inez et~al.(2019)Mart\'inez, Ahman, Dumitrescu, Giannarakis,
  Hawblitzel, Hri\cb{t}cu, Narasimhamurthy, Paraskevopoulou, Pit-Claudel,
  Protzenko, Ramananandro, Rastogi, and Swamy]{metafstar}
G.~Mart\'inez, D.~Ahman, V.~Dumitrescu, N.~Giannarakis, C.~Hawblitzel,
  C.~Hri\cb{t}cu, M.~Narasimhamurthy, Z.~Paraskevopoulou, C.~Pit-Claudel,
  J.~Protzenko, T.~Ramananandro, A.~Rastogi, and N.~Swamy.
\newblock \href {http://dx.doi.org/10.1007/978-3-030-17184-1\_2} {{Meta-F*}:
  Proof automation with {SMT}, tactics, and metaprograms}.
\newblock \iflongrefs{In \emph{28th European Symposium on Programming
  (ESOP)}}\else{\emph{ESOP}}\fi{}. \iflongrefs{april }\fi{}2019.

\bibitem[Matsuda et~al.(2023)Matsuda, Frohlich, Wang, and
  Wu]{embedding-by-unembedding}
K.~Matsuda, S.~Frohlich, M.~Wang, and N.~Wu.
\newblock \href {http://dx.doi.org/10.1145/3607830} {Embedding by unembedding}.
\newblock \emph{Proc. {ACM} Program. Lang.}, 7\penalty0 ({ICFP}):\penalty0
  1--47, 2023.

\bibitem[McBride(2010)]{mcbride-coincidences}
C.~McBride.
\newblock \href {http://dx.doi.org/10.1145/1863495.1863497} {Outrageous but
  meaningful coincidences: dependent type-safe syntax and evaluation}.
\newblock In B.~C. d.~S.~Oliveira and M.~Zalewski, editors, \emph{Proceedings
  of the {ACM} {SIGPLAN} Workshop on Generic Programming, {WGP} 2010,
  Baltimore, MD, USA, September 27-29, 2010}. 2010.

\bibitem[Meier et~al.(2025)Meier, Jensen, Pichon-Pharabod, and
  Spitters]{certicoq-wasm}
W.~Meier, M.~Jensen, J.~Pichon-Pharabod, and B.~Spitters.
\newblock \href {http://dx.doi.org/10.1145/3703595.3705879} {{CertiCoq-Wasm}: A
  verified {WebAssembly} backend for {CertiCoq}}.
\newblock In \emph{Proceedings of the 14th ACM SIGPLAN International Conference
  on Certified Programs and Proofs}. 2025.

\bibitem[Mullen et~al.(2018)Mullen, Pernsteiner, Wilcox, Tatlock, and
  Grossman]{oeuf}
E.~Mullen, S.~Pernsteiner, J.~R. Wilcox, Z.~Tatlock, and D.~Grossman.
\newblock \href {http://dx.doi.org/10.1145/3167089} {Œuf: minimizing the {Coq}
  extraction {TCB}}.
\newblock In \emph{Proceedings of the 7th ACM SIGPLAN International Conference
  on Certified Programs and Proofs}. 2018.

\bibitem[Myreen and Owens(2014)]{cakeml-pure-synthesis}
M.~O. Myreen and S.~Owens.
\newblock \href {http://dx.doi.org/10.1017/S0956796813000282} {Proof-producing
  translation of higher-order logic into pure and stateful {ML}}.
\newblock \emph{Journal of Functional Programming ({JFP})}, 24\penalty0
  (2-3):\penalty0 284--315, \iflongrefs{May }\fi{}2014.

\bibitem[Neis et~al.(2015)Neis, Hur, Kaiser, McLaughlin, Dreyer, and
  Vafeiadis]{NeisHKMDV15}
G.~Neis, C.~Hur, J.~Kaiser, C.~McLaughlin, D.~Dreyer, and V.~Vafeiadis.
\newblock \href {http://dx.doi.org/10.1145/2784731.2784764} {{Pilsner}: a
  compositionally verified compiler for a higher-order imperative language}.
\newblock \iflongrefs{In \emph{20th {ACM} {SIGPLAN} International Conference on
  Functional Programming}}\else{\emph{ICFP}}\fi{}, 2015.

\bibitem[New et~al.(2016)New, Bowman, and Ahmed]{NewBA16}
M.~S. New, W.~J. Bowman, and A.~Ahmed.
\newblock \href {http://dx.doi.org/10.1145/2951913.2951941} {Fully abstract
  compilation via universal embedding}.
\newblock \iflongrefs{In J.~Garrigue, G.~Keller, and E.~Sumii, editors,
  \emph{21st {ACM} {SIGPLAN} International Conference on Functional Programming
  ({ICFP})}}\else{\emph{ICFP}}\fi{}. 2016.

\bibitem[Nezamabadi et~al.(2026)Nezamabadi, Myreen, and
  Tan]{verified_dafny_extraction}
D.~Nezamabadi, M.~O. Myreen, and Y.~K. Tan.
\newblock \href {http://dx.doi.org/10.1145/3779031.3779092} {Verified {VCG} and
  verified compiler for {Dafny}}.
\newblock In \emph{Proceedings of the 15th ACM SIGPLAN International Conference
  on Certified Programs and Proofs}. 2026.

\bibitem[Patrignani et~al.(2019)Patrignani, Ahmed, and Clarke]{MarcosSurvey}
M.~Patrignani, A.~Ahmed, and D.~Clarke.
\newblock \href {http://dx.doi.org/10.1145/3280984} {Formal approaches to
  secure compilation: A survey of fully abstract compilation and related work}.
\newblock \emph{ACM Comput. Surv.}, 51\penalty0 (6), \iflongrefs{Feb.
  }\fi{}2019.

\bibitem[Pit{-}Claudel et~al.(2022)Pit{-}Claudel, Philipoom, Jamner, Erbsen,
  and Chlipala]{rupicola}
C.~Pit{-}Claudel, J.~Philipoom, D.~Jamner, A.~Erbsen, and A.~Chlipala.
\newblock \href {http://dx.doi.org/10.1145/3519939.3523706} {Relational
  compilation for performance-critical applications: extensible proof-producing
  translation of functional models into low-level code}.
\newblock In R.~Jhala and I.~Dillig, editors, \emph{{PLDI} '22: 43rd {ACM}
  {SIGPLAN} International Conference on Programming Language Design and
  Implementation, San Diego, CA, USA, June 13 - 17, 2022}. 2022.

\bibitem[Poiret et~al.(2026)Poiret, Maillard, and Tabareau]{poiret2026divide}
J.~Poiret, K.~Maillard, and N.~Tabareau.
\newblock \href {https://nantes-universite.hal.science/hal-05495420} {{Divide
  and Check: Logical Relations, No Algorithms Attached}}.
\newblock working paper or preprint, 2026.

\bibitem[Prinz et~al.(2022)Prinz, Kavvos, and
  Lampropoulos]{deeper-shallow-embeddings}
J.~Prinz, G.~A. Kavvos, and L.~Lampropoulos.
\newblock \href {http://dx.doi.org/10.4230/LIPICS.ITP.2022.28} {Deeper shallow
  embeddings}.
\newblock In J.~Andronick and L.~de~Moura, editors, \emph{13th International
  Conference on Interactive Theorem Proving, {ITP} 2022, Haifa, Israel, August
  7-10, 2022}. 2022.

\bibitem[Protzenko et~al.(2017)Protzenko, Zinzindohou\'e, Rastogi,
  Ramananandro, Wang, {Zanella-B\'eguelin}, Delignat-Lavaud, Hri\cb{t}cu,
  Bhargavan, Fournet, and Swamy]{lowstar}
J.~Protzenko, J.-K. Zinzindohou\'e, A.~Rastogi, T.~Ramananandro, P.~Wang,
  S.~{Zanella-B\'eguelin}, A.~Delignat-Lavaud, C.~Hri\cb{t}cu, K.~Bhargavan,
  C.~Fournet, and N.~Swamy.
\newblock \href {http://dx.doi.org/10.1145/3110261} {Verified low-level
  programming embedded in {F*}}.
\newblock \emph{{PACMPL}}, 1\penalty0 ({ICFP}):\penalty0 17:1--17:29,
  \iflongrefs{Sept. }\fi{}2017.

\bibitem[Protzenko et~al.(2020)Protzenko, Parno, Fromherz, Hawblitzel,
  Polubelova, Bhargavan, Beurdouche, Choi, Delignat{-}Lavaud, Fournet,
  Kulatova, Ramananandro, Rastogi, Swamy, Wintersteiger, and
  B{\'{e}}guelin]{evercrypt}
J.~Protzenko, B.~Parno, A.~Fromherz, C.~Hawblitzel, M.~Polubelova,
  K.~Bhargavan, B.~Beurdouche, J.~Choi, A.~Delignat{-}Lavaud, C.~Fournet,
  N.~Kulatova, T.~Ramananandro, A.~Rastogi, N.~Swamy, C.~M. Wintersteiger, and
  S.~Z. B{\'{e}}guelin.
\newblock \href {http://dx.doi.org/10.1109/SP40000.2020.00114} {{EverCrypt}:
  {A} fast, verified, cross-platform cryptographic provider}.
\newblock \iflongrefs{In \emph{2020 {IEEE} Symposium on Security and Privacy
  ({SP})}}\else{\emph{IEEE S{\&}P}}\fi{}. 2020.

\bibitem[Rastogi et~al.(2021)Rastogi, Mart\'inez, Fromherz, Ramananandro, and
  Swamy]{indexedeffects}
A.~Rastogi, G.~Mart\'inez, A.~Fromherz, T.~Ramananandro, and N.~Swamy.
\newblock \href {https://www.fstar-lang.org/papers/indexedeffects/}
  {Programming and proving with indexed effects}, \iflongrefs{July }\fi{}2021.

\bibitem[Sozeau et~al.(2025)Sozeau, Forster, Lennon-Bertrand, Nielsen,
  Tabareau, and Winterhalter]{metacoq2}
M.~Sozeau, Y.~Forster, M.~Lennon-Bertrand, J.~Nielsen, N.~Tabareau, and
  T.~Winterhalter.
\newblock \href {http://dx.doi.org/10.1145/3706056} {Correct and complete type
  checking and certified erasure for {Coq}, in {Coq}}.
\newblock \emph{J. ACM}, 72\penalty0 (1), \iflongrefs{Jan. }\fi{}2025.

\bibitem[Swamy(2026)]{VibeFStar}
N.~Swamy.
\newblock \href {https://risemsr.github.io/blog/2026-02-04-nik-agentic-pop/}
  {Agentic proof-oriented programming}.
\newblock RiSE MSR Blog, \iflongrefs{Feb. }\fi{}2026.

\bibitem[Swamy et~al.(2016)Swamy, Hri\cb{t}cu, Keller, Rastogi,
  Delignat-Lavaud, Forest, Bhargavan, Fournet, Strub, Kohlweiss,
  Zinzindohou\'e, and {Zanella-B\'eguelin}]{mumon}
N.~Swamy, C.~Hri\cb{t}cu, C.~Keller, A.~Rastogi, A.~Delignat-Lavaud, S.~Forest,
  K.~Bhargavan, C.~Fournet, P.-Y. Strub, M.~Kohlweiss, J.-K. Zinzindohou\'e,
  and S.~{Zanella-B\'eguelin}.
\newblock \href {http://dx.doi.org/10.1145/2837614.2837655} {Dependent types
  and multi-monadic effects in {F*}}.
\newblock \iflongrefs{In \emph{43rd ACM SIGPLAN-SIGACT Symposium on Principles
  of Programming Languages (POPL)}}\else{\emph{POPL}}\fi{}. \iflongrefs{Jan.
  }\fi{}2016.

\end{thebibliography}

\end{document}